\documentclass[12pt]{article}
\usepackage{amsmath}
\usepackage{graphicx,psfrag,epsf}
\usepackage{enumerate}
\usepackage{natbib}
\usepackage{url} 

\usepackage{hyperref}
\usepackage{amsmath,amsthm,amssymb}
\usepackage{mathtools}
\usepackage{color}
\usepackage{algpseudocode,algorithm}
\usepackage{array}
\usepackage{eqparbox}
\usepackage{float}
\usepackage{bm}
\RequirePackage{bbm}

\newtheorem{assumption}{Assumption}
\newtheorem{theorem}{Theorem}[section]
\newtheorem{lemma}{Lemma}[section]

\newtheorem{proposition}{Proposition}[section]
\newtheorem{remark}{Remark}

\newcommand{\argmin}{\operatornamewithlimits{arg\,min}}

\def \tr{\text{tr}}
\def \expect{\mathbb{E}}
\def \prob{\mathbb{P}}
\def \real{\mathbb{R}}
\def \diag{\text{diag}}
\DeclarePairedDelimiter\smallnorm{\lVert}{\rVert}
\providecommand{\norm}[1]{\left\lVert#1\right\rVert}
\def \({\left(}
\def \){\right)}
\def \[{\left[}
\def \]{\right]}

\def \M{M_0}

\def \U{U}
\def \V{V}
\def \Lam{\Lambda}
\def \lam{\lambda}

\def \Mf{M_F} 
\def \Mp{M} 
\def \Sig{\Sigma}
\def \y{y}
\def \p{p}
\def \sig{\sigma}
\def \eps{\epsilon}

\def \Mtilde{\widetilde{M}}

\newcommand{\balp}{\boldsymbol{\mathbf{\alpha}}}
\def \alphatilde{\widetilde{\balp}}
\def \phat{\hat{p}}

\def \Mhat{\hat{M}}
\def \lamhat{\hat{\lam}}
\def \Uhat{\hat{U}}
\def \Vhat{\hat{V}}
\def \s{s}
\def \shat{\hat{s}}

\def \n{n}
\def \d{d}
\def \rank{r}

\def \O{\mathcal{O}}
\def \Q{\mathcal{Q}}
\def \P{\mathcal{P}_{\Omega}}
\def \Portho{\mathcal{P}_{\Omega}^\perp}
\def \I{\mathbbm{1}} 

\def \bu{\mathbf{u}}
\def \bv{\mathbf{v}}
\newcommand{\blam}{\boldsymbol{\mathbf{\lambda}}}

\newcommand{\btau}{\boldsymbol{\mathbf{\tau}}}

\def \AI{\emph{Adaptive-Impute}}
\def \AIs{\emph{Adaptive-Impute }}
\def \SI{\emph{softImpute}}
\def \SIs{\emph{softImpute }}

\newcommand{\blind}{0}

\begin{document}

\def\spacingset#1{\renewcommand{\baselinestretch}%
{#1}\small\normalsize} \spacingset{1}


\if0\blind
{
  \title{\bf Intelligent Initialization and Adaptive Thresholding for Iterative Matrix Completion; Some Statistical and Algorithmic Theory for \AI}
  \author{Juhee Cho, Donggyu Kim, and Karl Rohe\thanks{
    This research is supported by NSF grant DMS-1309998 and ARO grant W911NF-15-1-0423.}\hspace{.2cm}\\
    Department of Statistics, University of Wisconsin-Madison\\
}
  \maketitle
} \fi

\if1\blind
{
  \bigskip
  \bigskip
  \bigskip
  \begin{center}
    {\LARGE\bf Intelligent Initialization and Adaptive Thresholding for Iterative Matrix Completion; Some Statistical and Algorithmic Theory for \AI}
\end{center}
  \medskip
} \fi

\bigskip
\begin{abstract}
Over the past decade, various matrix completion algorithms have been developed.
Thresholded singular value decomposition (SVD) is a popular technique in implementing many of them.
A sizable number of studies have shown its theoretical and empirical excellence, but
	choosing the right threshold level still remains as a key empirical difficulty.  
This paper proposes a novel matrix completion algorithm which iterates thresholded SVD 
	with theoretically-justified and data-dependent values of thresholding parameters. 
The estimate of the proposed algorithm enjoys the minimax error rate and shows outstanding empirical performances.
The thresholding scheme that we use can be viewed as a solution to a non-convex optimization problem, 
	understanding of whose theoretical convergence guarantee is known to be limited.
We investigate this problem by introducing a simpler algorithm, generalized-\SI, analyzing its convergence behavior, and 
	connecting it to the proposed algorithm.

\end{abstract}

\noindent%
{\it Keywords:}  \SI, generalized-\SI, non-convex optimization, thresholded singular value decomposition  
\vfill

\newpage
\spacingset{1.45} 
\section{Introduction} \label{intro2}
Matrix completion appears in a variety of areas where it recovers a low-rank or approximately low-rank matrix from a small fraction of observed entries 
	such as collaborative filtering (\cite{rennie2005}), computer vision (\cite{weinberger2006}), positioning (\cite{montanari2010}), and 
		recommender systems (\cite{bennett2007}).
Early work in this field was done by \cite{achlioptas2001}, \cite{azar2001}, \cite{fazel2002}, \cite{srebro2004}, and \cite{rennie2005}.  
Later, \cite{candes2009recht} introduced the technique of matrix completion 
	by minimizing the nuclear norm under convex constraints.
This opened up a significant overlap with compressed sensing (\cite{candes2006}, \cite{donoho2006}) and 
	led to accelerated research in matrix completion. 
They and others (\cite{candes2009recht}, \cite{candes2010tao}, \cite{keshavan2010a}, \cite{gross2011}, \cite{recht2011}) showed that the technique can exactly recover a low-rank matrix in the noiseless case.  
Many of the following works showed the approximate recovery of the low-rank matrix with the presence of noise (\cite{candes2010plan}, \cite{negahban2011}, \cite{koltchinskii2011}, \cite{rohde2011}).
Several other papers studied matrix completion in various settings (e.g. \cite{davenport2014}, \cite{negahban2012}) and 
proposed different estimation procedures of matrix completion (\cite{srebro2004}, \cite{keshavan2009}, \cite{koltchinskii2011solo}, \cite{cai2013}, \cite{chatterjee2014}) than the ones by \cite{candes2009recht}.
In addition to the theoretical advances, a large number of algorithms have emerged (e.g. \cite{rennie2005}, \cite{cai2010}, \cite{keshavan2009}, \cite{mazumder2010}, \cite{hastie2014}).  An overview is well summarized in \cite{mazumder2010} and \cite{hastie2014}. 


Many of matrix completion algorithms employ thresholded singular value decomposition (SVD) which soft- or hard- thresholds the singular values.
The statistical literature has responded by investigating its theoretical optimality and strong empirical performances. 
However, a key empirical difficulty of employing thresholded SVD for matrix completion is to find the right way and level of threshold.
Depending on the choice of the thresholding scheme, 
	the rank of the estimated low-rank matrix and predicted values for unobserved entries can widely change. 
Despite its importance, we lack understanding on how to choose the threshold level and what bias or error we eliminate by thresholding.

We propose a novel iterative matrix completion algorithm, \AI, which 
	recovers the underlying low-rank matrix from a few noisy entries via differentially and adaptively thresholded SVD. 
Specifically, the proposed \AIs algorithm differentially thresholds the singular values and adaptively updates the threshold levels on every iteration. 
As was the case with adaptive Lasso (\cite{zou2006}) and adaptive thresholding for sparse covariance matrix estimation (\cite{cai2011adaptive}),
	the proposed thresholding scheme gives \AIs stronger empirical performances than the thresholding scheme 
	that uses a single thresholding parameter for all singular values throughout the iterations (e.g. \SIs (\cite{mazumder2010})). 
Although \AIs employs multiple thresholding parameters changing over iterations, 
	we suggest specified values for the thresholding parameters that are theoretically-justified and data-dependent.
Hence, \AIs is free of the tuning problems associated with the choice of threshold levels. 
Its single tuning parameter is the rank of the resulting estimator. 
We suggest a way to choose the rank based on singular value gaps (for details, see Section \ref{realdata}).
This novel threshold scheme of \AIs makes it estimation via non-convex optimization, 
	understanding of whose theoretical guarantees is known to be limited. 
However, to solve this problem and help understand the convergence behavior of \AI, 
	we introduce a simpler algorithm than \AI, generalized-\SI, and derive a sufficient condition under which it converges.
Then, we prove that \AIs behaves almost the same as generalized-\SI.
Numerical experiments and a real data analysis in Section \ref{sim} suggest superior performances of \AIs over the existing \SI-type algorithms.

The rest of this paper is organized as follows.
Section \ref{setup} describes the model setup. 
Section \ref{AI} introduces the proposed algorithm \AI. 
Section \ref{generalSI} introduces a generalized-\SI, a simpler algorithm than \AI.
Section \ref{sim} presents numerical experiment results. 
Section \ref{discuss} concludes the paper with discussion. 
All proofs are collected in Section \ref{proofs2}.

\section{The model setup} \label{setup}
Suppose that we have an $\n\times\d$ matrix of rank $\rank$,
\begin{equation} \label{mod}
\M = \U\Lam\V^T,
\end{equation}
where by SVD, 
	$\U=(\U_1,\ldots,\U_\rank) \in \real^{\n \times \rank}$, 
	$\V=(\V_1,\ldots,\V_\rank) \in \real^{\d \times \rank}$, 
	$\Lam = \diag (\lam_1, \ldots, \lam_\rank)\\ \in \real^{\rank \times \rank}$, and
	$\lam_1\ge \ldots\ge \lam_\rank\ge0$.
The entries of $\M$ are corrupted by noise $\eps \in \real^{\n \times \d}$ whose entries are i.i.d. sub-Gaussian random variables with mean zero and variance $\sig^2$. 
Hence, we can only observe $\Mf = \M + \eps$. 
However, oftentimes in real world applications, not all entries of $\Mf$ are observable.
So, define $y \in \real^{\n \times \d}$ such that $y_{ij} = 1$ if the $(i,j)$-th entry of $\Mf$ is observed and $y_{ij} = 0$ if it is not observed.  
The entries of $y$ are assumed to be i.i.d. Bernoulli($p$) and independent of the entries of $\epsilon$.   
Then, the partially-observed noisy low-rank matrix $\Mp \in \real^{\n \times \d}$ is written as
\begin{align*}
\Mp_{ij} = \y_{ij} {\Mf}_{ij}
		=  \left\{\begin{matrix}
				{\M}_{ij}+\eps_{ij} & \text{if observed ($\y_{ij}=1$)}\\ 
				0 &  \text{otherwise ($\y_{ij}=0$).} 
				\end{matrix}\right. 
\end{align*}
Throughout the paper, we assume that $r \ll d \le \n$ and the entries of $\M$ are bounded by a positive constant $L$ in absolute value.
In this paper, we develop an iterative algorithm to recover $\M$ from $\Mp$ and investigate its theoretical properties and empirical performances.

\section{\AIs algorithm} \label{AI}

\subsection{Initialization} \label{initial}

We first introduce some notation.
Let a set $\Omega$ contain indices of the observed entries, 
$y_{ij}=1 \Leftrightarrow (i,j) \in \Omega.$
Then, for any matrix $A \in \real^{\n\times\d}$, denote by $\P(A)$ the projection of $A$ onto $\Omega$ and by $\Portho(A)$ the projection of $A$ onto the complement of $\Omega$;
$$\[\P(A)\]_{ij} = \left\{\begin{matrix}
A_{ij} & \text{if } (i,j) \in \Omega\\ 
0 & \text{if } (i,j) \notin \Omega
\end{matrix}\right.
\quad\text{and}\quad
\[\Portho(A)\]_{ij} = \left\{\begin{matrix}
0 & \text{if } (i,j) \in \Omega\\ 
A_{ij} & \text{if } (i,j) \notin \Omega.
\end{matrix}\right.$$
That is, $\P(A) + \Portho(A) = A$.
We let $\bu_i(A)$ denote the $i$-th left singular vector of $A$, 
	$\bv_i(A)$ the $i$-th right singular vector of $A$, and $\blam_i(A)$ the $i$-th singular value of $A$ 
	such that $\blam_1(A)\ge \ldots\ge \blam_\d(A)$. 
The squared Frobenius norm is defined by $\left \| A \right \| _F ^2 = \tr\(A^T A\)$, the trace of $A^T A$, and
	the nuclear norm by $\left \| A \right \| _\ast = \sum_{i=1}^\d \blam_i(A)$, the sum of the singular values of $A$.
For a symmetric matrix $A \in \real^{\n\times\n}$, diag(A) represents a matrix with diagonal elements of A on the diagonal and zeros elsewhere. 

Many of the iterative matrix completion algorithms (e.g. \cite{cai2010}, \cite{mazumder2010}, \cite{keshavan2009}, \cite{chatterjee2014}) in the current literature initialize with $\Mp$, where the unobserved entries begin at zero. 
This initialization works well with algorithms that are based on convex optimization or that are robust to the initial. 
However, for algorithms that are based on non-convex optimization or that are sensitive to the initial, filling the unobserved entries with zeros may not be a good choice. 
\cite{cho2015} proposed a one-step consistent estimator, $\Mhat$, 
	that attains the minimax error rate (\cite{koltchinskii2011}), $\rank/\p\d$, and requires only two eigendecompositions. 
\AIs employs the entries of this one-step consistent estimator instead of zeros as initial values of the unobserved entries.
Algorithm \ref{alg:initial} describes how to compute the initial $\Mhat$ of \AI.
The following theorem shows that $\Mhat$ achieves the minimax error rate.

\begin{algorithm}[h]
\caption{\;Initialization (\cite{cho2015})}
\begin{algorithmic}
    \Require{$\Mp$, $\y$, and $r$}
    \State 
        $\phat \gets \frac{1}{\n\d}\sum_{i=1}^\n \sum_{j=1}^\d \y_{ij}$
    \State 
        ${\Sig}_{\phat} \gets \Mp^T\Mp - (1-\phat) \diag(\Mp^T\Mp)$
    \State 
        ${\Sig}_{t \phat} \gets \Mp\Mp^T - (1-\phat) \diag(\Mp\Mp^T)$
    \State 
        $\Vhat_i \gets \bv_i({\Sig}_{\phat}), \quad \forall i \in \{1,\ldots,\rank\}$
    \State 
        $\Uhat_i \gets \bu_i({\Sig}_{t \phat}), \quad \forall i \in \{1,\ldots,\rank\}$
    \State 
        $\alphatilde \gets \frac{1}{\d-\rank} \sum_{i=\rank+1}^\d \blam_i({\Sig}_{\phat})$
    \State 
        $\hat\tau_i \gets \blam_i({\Sig}_{\phat}) - \frac{1}{\phat}\sqrt{\blam_i({\Sig}_{\phat}) - \alphatilde}, \quad \forall i \in \{1,\ldots,\rank\}$
    \State 
        $\lamhat_i \gets \blam_i({\Sig}_{\phat}) - \hat\tau_i, \quad\quad\quad\quad\quad\quad\;\;\, \forall i \in \{1,\ldots,\rank\}$
    \State 
        $\shat=(\shat_1,\ldots,\shat_\rank) \gets \argmin_{s \in \{-1,1\}^\rank} \norm{ \P\( \sum_{i=1}^\rank s_i \lamhat_i \Uhat_i \Vhat_i^T - \Mp \) }_F^2$
    \State $\Mhat \gets \sum_{i=1}^\rank \shat_i \lamhat_i \Uhat_i \Vhat_i^{T}$
    \\
    \Return{$\Mhat$}
\end{algorithmic}
\label{alg:initial}
\end{algorithm}

\begin{assumption}\label{assume1}
~
\begin{enumerate}
	\item [(1)] $\p\d/\log\n \to \infty$ and $\n,\d \to \infty$ with $\;\d \le \n \le e^{\d^{\beta}},\,$ where $\;\beta<1$ free of $\n$, $\d$, and $\p$;
	\item [(2)] $\lam_{i} = b_{i} \sqrt {\n\d}$ for all $i =1, \ldots,\rank$, where $\{b_{i}\}_{i =1, \ldots,\rank}$ are positive bounded values;
	\item [(3)] $b_{i} > b_{i+1}$ for all $i =1, \ldots,\rank,$ where $b_{\rank+1}=0$;
	\item [(4)] $\lim_{\n,\d\to\infty} \prob\Big(\min_{\s \in \{-1,1\}^\rank} \; 
	\big\| \P\big( \sum_{i=1}^\rank s_i \lamhat_i \Uhat_i \Vhat_i^T - \Mp \big) \big\|_F^2$ 
	\\ \hspace*{7.8cm}	$< \big\| \P\big( \sum_{i=1}^\rank s_{0i} \lamhat_i \Uhat_i \Vhat_i^T - \Mp \big) \big\|_F^2 \Big) = 0$,
	\\ where $\s=(\s_1,\ldots,\s_\rank)$ and $\s_{0i} = \text{sign}(\langle \Vhat_i, \V_i \rangle) \;\text{sign}(\langle \Uhat_i, \U_i \rangle)$ for $i =1, 		
\ldots,\rank$.
\end{enumerate}
\end{assumption}
\begin{remark}
Under the setting where the rank $\rank$ is fixed as in this paper, 
	Assumption \ref{assume1}(2) implies that the underlying low-rank matrix $\M$ is dense.
More specifically, 
	note that the squared Frobenius norm indicates both the sum of all squared entries of a matrix and the sum of its singular values squared. 
Also, note that $\smallnorm{\M}_F^2 = \sum_{i=1}^{\rank}\blam_i^2(\M) = c \n\d$ for some constant $c>0$ by Assumption \ref{assume1}(2).
Thus, the sum of all squared entries of $\M$ has an order $\n\d$.
This means that a non-vanishing proportion of entries of $\M$ contains non-vanishing signals with dimensionality (see \cite{fan2013}).
For more discussion, see Remark 2 in \cite{cho2015}.
\end{remark}
\begin{remark}
The singular vectors, $\{\Uhat_i\}_{i=1}^\rank$ and $\{\Vhat_i\}_{i=1}^\rank$, that compose $\Mhat$ 
	are consistent estimators of $\U$ and $\V$ up to  signs (for details, see \cite{cho2015}). 
Hence, when combining them with $\{\lamhat_i\}_{i=1}^\rank$ to reconstruct $\Mhat$, a sign problem happens. 
Assumption \ref{assume1}(4) assures that as $\n$ and $\d$ increase, 
	the probability of choosing different signs than the true signs, $\{\s_{0i}\}_{i=1}^\rank$, goes to zero. 
Given the asymptotic consistency of $\{\Uhat_i\}_{i=1}^\rank$, $\{\Vhat_i\}_{i=1}^\rank$, and $\{\lamhat_i\}_{i=1}^\rank$, 
	this is not an unreasonable assumption to make.
\end{remark}

\begin{proposition}(Theorem 4.4 in \cite{cho2015}) \label{MhatConsistent}
Under Assumption \ref{assume1} and the model setup in Section \ref{setup}, 
	$\Mhat$ is a consistent estimator of $\M$.  In particular,
\begin{equation*}
\frac{1}{\n\d}\smallnorm{ \Mhat - \M }_F^2 = o_p\(\frac{h_\n }{\p\d}\)\,,
\end{equation*}
where $h_\n$ diverges very slowly with the dimensionality, for example, $\log( \log \d)$.
\end{proposition}
\begin{remark} \label{minimaxRate}
Since $h_\n$ in Proposition \ref{MhatConsistent} can be any quantity that diverges slowly with the dimensionality, the convergence rate of $\Mhat$ can be thought of as $1/\p\d$. Under the setting where the rank of $\M$ is fixed as in this paper, it is matched to the minimax error rate, $\rank/\p\d$, found in \cite{koltchinskii2011}. 
\end{remark}

Using $\Mhat$ to initialize \AIs has two major advantages. 
First, since $\Mhat$ is already a consistent estimator of $\M$ achieving the minimax error rate, 
	it allows a series of the iterates of \AIs coming after $\Mhat$ to be also consistent estimators of $\M$ achieving the minimax error rate (see Theorem \ref{thm:mathInduct}). 
Second, because \AIs is based on a non-convex optimization problem (see Section \ref{generalSI}), its convergence may depend on initial values. 
$\Mhat$ provides \AIs a suitable initializer. 

\subsection{Adaptive thresholds} \label{adapt-thresh}


To motivate the novel thresholding scheme of \AI, we first consider the case where a fully-observed noisy low-rank matrix is available.
Specifically, suppose that the probability of observing each entry, $\p$, is $1$ and thus $\Mf = \M+\eps$ is observed.
Under the model setup in Section \ref{setup} we can easily show that
\begin{equation} \label{expect-Mf}
\expect (\Mf^T \Mf) = \M^T \M + \n\sig^2 I_\d
\quad\text{and}\quad
\expect (\Mf \Mf^T) = \M \M^T + \d\sig^2 I_\n,
\end{equation}
where $I_{\d}$ and $I_{\n}$ are identity matrices of size $\d$ and $\n$, respectively.
This shows that the eigenvectors of $\expect (\Mf^T \Mf)$ and $\expect (\Mf \Mf^T)$ are the same as the right and left singular vectors of $\M$.
Also, the top $\rank$ eigenvalues of $\expect (\Mf^T \Mf)$ consist of the squared singular values of $\M$ and a noise, $\n\sig^2$, 
	the latter of which is the same as the average of the bottom $\d-\rank$ eigenvalues of $\expect (\Mf^T \Mf)$.
In light of this, we want the estimator of $\M$ based on $\Mf$ to keep the first $\rank$ singular vectors of $\Mf$ as they are, but adjust the bias occuring in the singular values of $\Mf$.
Thus, the resulting estimator is
\begin{equation} \label{estor-M-fully}
\Mhat^F = \sum_{i=1}^\rank \sqrt{\blam_i^2 (\Mf)-\balp} \; \bu_i (\Mf) \bv_i (\Mf)^T,
	\quad\text {where }
	\balp = \frac{1}{\d-\rank}\sum_{i=\rank+1}^\d \blam_i^2 (\Mf).
\end{equation}
A simple extension of Proposition \ref{MhatConsistent} shows that 
	$\Mhat^F$ achieves the best possible minimax error rate of convergence, $1/\d$, since $\p=1$.

Now consider the cases where a partially-observed noisy low-rank matrix $\Mp$ is available.
For each iteration $t\ge 1$, we fill out the unobserved entries of $\Mp$ with the corresponding entries of the previous iterate $Z_{t}$,
		treat the completed matrix $\Mtilde_{t}=\P(\Mp)+\Portho(Z_{t})$ as if it is a fully-observed matrix $\Mf$,
		and find the next iterate $Z_{t+1}$ in the same way that we found $\Mhat^F$ from $\Mf$ in \eqref{estor-M-fully};
\begin{equation} \label{estor-Mhat-t}
Z_{t+1} = \sum_{i=1}^\rank \sqrt{\blam_i^2 (\Mtilde_{t})-\alphatilde_{t}} \; \bu_i (\Mtilde_{t}) \bv_i (\Mtilde_{t})^T,
	\quad\text {where }
	\alphatilde_{t} = \frac{1}{\d-\rank}\sum_{i=\rank+1}^\d \blam_i^2 (\Mtilde_{t}).
\end{equation}
Note that the difference in \eqref{estor-Mhat-t} from \eqref{estor-M-fully} is in the usage of $\Mtilde_{t}$ instead of $\Mf$. 
Hence, the performance of \AIs may depend on how close $\P(Z_t)$ is to $\P(\M)$.
Algorithm \ref{alg:AI} summarizes these computing steps of \AIs continued from Algorithm \ref{alg:initial}.

\begin{algorithm}[ht]
\caption{\;\AI}
\begin{algorithmic}
    \Require{$\Mp$, y, $r$, and $\varepsilon > 0$}
    \State $Z_1 \gets \Mhat$          \Comment{from Algorithm \ref{alg:initial}}
    \Repeat \;\; for $t=1,2,\ldots$
    \State 
        $\Mtilde_{t} \gets \P(\Mp)+\Portho(Z_{t})$
    \State
        $V_i^{(t)} \gets \bv_i (\Mtilde_t), \hspace*{0.5cm} \forall i \in \{1,\ldots,\rank\}$
    \State
        $U_i^{(t)} \gets \bu_i (\Mtilde_t), \hspace*{0.5cm} \forall i \in \{1,\ldots,\rank\}$
    \State
        $\alphatilde_{t} \gets \frac{1}{\d-\rank} \sum_{i=\rank+1}^{\d} \blam_{i}^2 (\Mtilde_{t})$
    \State
        $\tau_{t,i} \gets \blam_i (\Mtilde_t) - \sqrt{\blam_i^2 (\Mtilde_t) - \alphatilde_{t}}, \hspace*{1.9cm} \forall i \in \{1,\ldots,\rank\}$ \Comment{Adaptive thresholds}
    \State
        $\lam_i^{(t)} \gets \blam_i (\Mtilde_t) - \tau_{t,i} \(= \sqrt{\blam_i^2 (\Mtilde_t) - \alphatilde_{t}}\), \quad \forall i \in \{1,\ldots,\rank\}$
    \State
        $Z_{t+1} \gets \sum_{i=1}^\rank \lam_i^{(t)} U_i^{(t)} V_i^{(t)T}$
    \State
        $t \gets t + 1$
    \Until{$\norm{Z_{t+1}-Z_{t}}_F^2/\norm{Z_{t}}_F^2 \le \varepsilon$}
    \\
    \Return{$Z_{t+1}$}
\end{algorithmic}
\label{alg:AI}
\end{algorithm}

The following theorem illustrates that the iterates of \AIs retain the statistical performance of the initializer $\hat M$.
\begin{assumption} \label{assume3}
For all $i =1, \ldots,\rank, \;\text{sign}(\langle \bu_i(\Mtilde_t),\U_i \rangle)=\text{sign}(\langle \bv_i(\Mtilde_t),\V_i \rangle).$ 
\end{assumption}

\begin{theorem} \label{thm:mathInduct}
Under Assumptions \ref{assume1}-\ref{assume3} and the model setup in Section \ref{setup}, we have for any fixed value of $t$, 
$$\frac{1}{\n\d}\norm{Z_t -\M}_F^2 = o_p\(\frac{h_\n}{\p\d}\),  \ \mbox{ as $n,d \rightarrow \infty$ with any }  h_\n \to \infty$$
where $h_\n$ diverges very slowly with the dimensionality, for example, $\log( \log \d)$.
\end{theorem}
\begin{remark}
Similarly as in Remark \ref{minimaxRate}, since $h_\n$ is a quantity diverging very slowly, 
	the convergence rate of $Z_t$ can be thought of as $1/\p\d$ which is matched to the minimax error rate, $\rank/\p\d$ (\cite{koltchinskii2011}). 
\end{remark}

\subsection{Non-convexity of \AI}
We can view \AIs as an estimation method via non-convex optimization. 

For $t\ge1$, define
\begin{equation} \label{taudef}
\tau_{t,i}= 
\left\{\begin{array}{ll}
\blam_i (\Mtilde_t) - \sqrt{\blam_i^2 (\Mtilde_t) - \alphatilde_{t}}, & i \le \rank \\
\blam_{\rank+1} (\Mtilde_t),  & i >\rank\end{array}\right.,
\end{equation}
where $\alphatilde_{t} = \frac{1}{\d-\rank} \sum_{i=\rank+1}^{\d} \blam_{i}^2 (\Mtilde_{t})$ and $\Mtilde_t = \P\(\Mp\) + \Portho\(Z_t\)$.
Then, in each iteration \AIs provides a solution to the problem 
\begin{equation} \label{eq:Qt}
\min_{Z\in\real^{\n\times\d}} \; \frac{1}{2\n\d}\,\smallnorm{\Mtilde_t - Z}_F^2 
	+ \sum_{i=1}^\d \frac{\tau_{t,i}}{\sqrt{\n\d}} \frac{\blam_i (Z)}{\sqrt{\n\d}}.
\end{equation}
Note that the threshold parameters, $\tau_{t,i}$, have dependence on both the $i$-th singular value and the $t$-th iteration. 
The following theorem provides an explicit solution to \eqref{eq:Qt}.

\begin{theorem} \label{solutionZ}
Let $X$ be an $\n\times\d$ matrix and let $\n\ge\d$. The optimization problem 
\begin{eqnarray} \label{ProblemZ}
\min_Z \;\frac{1}{2\n\d}\norm{X-Z}_F^2 + \sum_{i=1}^\d \frac{\tau_i}{\sqrt{\n\d}} \frac{\blam_i(Z)}{\sqrt{\n\d}}
\end{eqnarray} 
has a solution which is given by 
\begin{eqnarray} \label{Zhat}
\hat Z = \Phi (\Delta-\btau)_+\Psi ^T,
\end{eqnarray}
where $\Phi\Delta\Psi^T$ is the SVD of $X$, $\btau = \diag(\tau_1,\ldots,\tau_\d)\in\real^{\d\times\d}$, $(\Delta-\btau)_+ = \diag\big( (\blam_1(X)-\tau_1)_+,\ldots, (\blam_\d(X)-\tau_\d)_+\big)\in\real^{\d\times\d}$, and $c_+ = \max(c,0)$ for any $c\in\real$.
\end{theorem}

\begin{remark} \label{defineTau}
To see how Theorem \ref{solutionZ} provides a solution to \eqref{eq:Qt}, 
	let $X = \Mtilde_t$ and $\tau_i=\tau_{t,i}$ as specified in \eqref{taudef}.
Then, \eqref{eq:Qt} and \eqref{ProblemZ} become the same and $\hat Z$ in \eqref{Zhat} gives the explicit form of the $(t+1)$-th iterate, $Z_{t+1}$, in Algorithm \ref{alg:AI}. 
\end{remark}

If all of the thresholding parameters in \eqref{eq:Qt} are equal 
	such that $\tau=\tau_{t,1}=\ldots=\tau_{t,\d}$ for all $1 \le i \le d$ and $t\ge 1$, 
	the optimization problem \eqref{eq:Qt} becomes equivalent to that of \SIs (\cite{mazumder2010}) and 
	Theorem \ref{solutionZ} provides an iterative solution to it. 
While \SIs requires finding the right value of a thresholding parameter $\tau$ by using a cross validation (CV) technique 
	which is time-consuming and often does not have a straightforward validation criteria, 
	\AIs suggests specific values of the thresholding levels as in \eqref{taudef}.
The novel thresholding scheme of \AIs together with the rank constraint 
	results in superior empirical performances over the existing \SI-type algorithms (see Section \ref{sim}). 

The thresholding scheme of \AIs can be viewed as a solution to a non-convex optimization problem
	since at every iteration it differentially and adaptively thresholds the singular values.
As Hastie and others alluded to a similar issue for matrix completion methods via non-convex optimization in \cite{hastie2014}, 
	it is hard to provide a direct convergence guarantee of \AI. 
So, in the following section we introduce a generalized-\SIs algorithm, simpler than \AIs and yet still non-convex, 
	and investigate its asymptotic convergence. 
It hints at the convergent behavior of \AIs in the asymptotic sense. 

\section{Generalized \SI} \label{generalSI}

Generalized-\SIs is an algorithm which iteratively solves the problem,
\begin{equation} \label{eq:Q}
\min_{Z\in\real^{\n\times\d}} \; Q_{\tau}(Z|Z_{t}^g) := \frac{1}{2\n\d}\norm{\P\(\Mp\) + \Portho\(Z_t^g\) - Z}_F^2 
	+ \sum_{i=1}^\d \frac{\tau_{i}}{\sqrt{\n\d}} \frac{\blam_i (Z)}{\sqrt{\n\d}}, 
\end{equation}
to ultimately solve the optimization problem,
\begin{equation} \label{eq:f}
 \min_{Z\in\real^{\n\times\d}} \; f_{\tau}(Z) := \frac{1}{2\n\d}\norm{\P(\Mp) - \P(Z)}_F^2 
	+ \sum_{i=1}^\d \frac{\tau_{i}}{\sqrt{\n\d}} \frac{\blam_i (Z)}{\sqrt{\n\d}}.
\end{equation}
Note that generalized-\SIs differentially penalizes the singular values, but the thresholding parameters do not change over iterations.
The iterative solutions of generalized-\SIs are denoted by $Z^{g}_{t+1}:= \argmin_{Z\in\real^{\n\times\d}} Q_{\tau}(Z|Z_{t}^g)$ for $t\ge 1$ and 
	Theorem \ref{solutionZ} provides a closed form of $Z^{g}_{t+1}$.
If $\tau_i = \tau$ for all $1\le i\le \d$, generalized-\SIs will be equivalent to \SIs and both \eqref{eq:Q} and \eqref{eq:f} become convex problems.
However, by differentially penalizing the singular values, generalized-\SIs ends up solving a non-convex optimization problem.
Theorem \ref{thm:Zsolvesf} below shows that despite the non-convexity of generalized-\SI,
	the iterates of generalized-\SI, $\{Z^g_t\}_{t\ge 1}$, converge to a solution of problem \eqref{eq:f} under certain conditions. 

\begin{assumption} \label{assume2}
Let $\Mtilde_{t}^g = \P(\Mp) + \Portho(Z_t^g)$ and $D_t^g := \Mtilde_{t}^g - Z^g_{t+1}$. Then,
$$\frac{1}{nd}\norm{D_t^g - D_{t+1}^g}_F^2 +\frac{2}{nd}\langle D_t^g - D_{t+1}^g,Z^g_{t+1} - Z^g_{t+2} \rangle \ge 0 \quad\text{for all }\; t\ge1.$$
\end{assumption}

\begin{theorem} \label{thm:Zsolvesf}
Let $Z_\infty$ be a limit point of the sequence $Z_t^g$.
Under Assumption \ref{assume2}, if the minimizer $Z^s$  of \eqref{eq:f} satisfies
\begin{eqnarray}\label{assume4}
Z^s \in  \bigg \{ Z \in \mathbb{R}^{n \times d} : \sum_{i=1}^d \tau_i \lam_i (Z) \geq  \sum_{i=1}^d \tau_i \lam_i (Z_\infty) +  \langle (Z-Z_\infty), D_\infty  \rangle \bigg  \},
\end{eqnarray}
we have $f_{\tau}(Z_{\infty}) = f_{\tau}(Z ^s)$ and $\lim_{t\rightarrow \infty} f_{\tau}(Z^g_{t}) = f(Z^s)$.
\end{theorem}
\begin{remark} \label{connectTOsoftimpute}
If $\tau_i = \tau$ for all $i$ as in case of \SI, Assumption \ref{assume2} and \eqref{assume4} are always satisfied because
$\frac{1}{\tau}D_t^g$ belongs to the sub-gradient of $\norm{Z_{t+1}^g}_\ast$ .
\end{remark}
\begin{remark}
If $Z^s$ is unique, then 
	generalized-\SIs finds the global minimum point of \eqref{eq:f} by Theorem \ref{thm:Zsolvesf}. 
\end{remark}

Generalized-\SIs resembles \AIs in a sense that both of them employ different thresholding parameters on $\blam_i (Z)$'s.  
However, \AIs updates these tuning parameters every iteration while generalized-\SIs does not. 
The following lemmas show that despite this difference, 
	the convergent behavior of \AIs is asymptotically close to that of generalized-\SI. 

\begin{lemma} \label{lem:undateTuning}
Under Assumptions \ref{assume1}-\ref{assume3} and the model setup in Section \ref{setup}, we have
\begin{eqnarray*}
\left| \frac{\tau_{t,i}}{\sqrt{\n\d}} - \frac{\tau_{t+1,i}}{\sqrt{\n\d}} \right|
	= o_p\(\sqrt{\frac{h_\n}{\p\d}}\) \quad\text{for }\; i=1,\ldots,\d,
\end{eqnarray*}
where $\tau_{t,i}$ is defined in \eqref{taudef}.
\end{lemma}

\begin{lemma} \label{corol:assume2}
Let $D_t := \Mtilde_{t} - Z_{t+1}$, where $\Mtilde_{t}$ and $Z_{t}$ are as defined in Algorithm \ref{alg:AI}.
Then, under Assumptions \ref{assume1}-\ref{assume3} and the model setup in Section \ref{setup}, we have
$$\frac{1}{nd}\norm{D_t - D_{t+1}}_F^2 +\frac{2}{nd}\langle D_t - D_{t+1},Z_{t+1} - Z_{t+2} \rangle + o_p\(\frac{h_\n}{\p\d}\) \ge 0.$$
\end{lemma}

Lemma \ref{lem:undateTuning} shows that for large $\n$ and $\d$, 
	thresholding parameters of \AIs are stable between iterations so that \AIs behaves similarly to generalized-\SI.
Lemma \ref{corol:assume2} shows how Assumption \ref{assume2} is adapted in \AI.
It implies a possibility of \AIs satisfying Assumption \ref{assume2} asymptotically. 
Although this still does not provide a guarantee of convergence of \AI, numerical results below support this possibility. 

\section{Numerical results} \label{sim}

In this section, we conducted simulations and a real-data analysis to compare \AIs for estimating $\M$ with the four different versions of \SI:
\begin{enumerate}
\item \AI: the proposed algorithm, as summarized in Algorithm \ref{alg:AI};
\item  \SI: the original \SIs algorithm (\cite{mazumder2010});
\item \emph{softImpute-Rank}: \SIs with rank restriction (\cite{hastie2014});
\item \emph{softImpute-ALS}: \emph{Maximum-Margin Matrix Factorization} (\cite{hastie2014}); 
\item \emph{softImpute-ALS-Rank}: \emph{rank-restricted Maximum-Margin Matrix Factorization} in Algorithm 3.1 (\cite{hastie2014}).
\end{enumerate}
\emph{SoftImpute} algorithms were implemented with the \texttt{R} package, \texttt{softImpute} (\cite{hastie2015}). 
The R code for \AIs is available at \url{https://github.com/chojuhee/hello-world/blob/master/adaptiveImpute_Rfunction}.
In this R code, we made two adjustments from Algorithms \ref{alg:initial} and \ref{alg:AI} for technical reasons. 
First, in almost all real world applications that needed matrix completion, the entries of $\M$ are bounded below and above by constants $L_1$ and $L_2$ such that
$$ L_1 \le {\M}_{ij} \le L_2 \ $$
	and smaller or larger values than the constants do not make sense. 
So, after each iteration of \AI, $t\ge1$, we replace 
	the values of $Z_t$ that are smaller than $L_1$ with $L_1$ and  
	the values of $Z_t$ that are greater than $L_2$ with $L_2$. 
Second, the cardinality of the set, $\{-1,1\}^\rank$, that we search over to find $\shat$ in Algorithm \ref{alg:initial} increases exponentially. 
Hence, finding $\shat$ easily becomes a computational bottleneck of \AIs or is even impossible for large $\rank$.
We suggest two possible solutions to this problem. 
One solution is to find $\shat$ by computing
	$\shat_{i} = \text{sign}(\langle \Vhat_i, \bv_i(\Mp) \rangle) \; \text{sign}(\langle \Uhat_i, \bu_i(\Mp) \rangle) \text{ for } i=1,\ldots,\rank$. 
Note that if we use $\V_i$ and $\U_i$ instead of $\bv_i(\Mp)$ and $\bu_i(\Mp)$, this gives us the true sign $s_0$ under Assumption \ref{assume1}. 
The other solution is to use a linear regression. 
Let a vector of the observed entries of $\Mp$ be the dependent variable and
	let a vector of the corresponding entries of $\lamhat_i \Uhat_i \Vhat_i^T$ be the $i$-th column of the design matrix for $i=1,\ldots,\rank$.
Then, we set $\shat$ to be the coefficients of the regression line whose intercept is forced to be 0.
The difference in the results of these two methods are negligible.
In the following experiment, we only reported the results of the former solution for simplicity,
	while the R code provided in \url{https://github.com/chojuhee/hello-world/blob/master/adaptiveImpute_Rfunction} are written for both solutions.

\subsection{Simulation study}

To create $\M=A B^{T} \in \real^{\n\times\d}$, 
	we sampled $A \in \mathbb{R}^{\n \times r}$ and $B \in  \mathbb{R}^{\d \times r}$ to contain i.i.d. uniform$[-5,5]$ random variables  
	and a noise matrix $\epsilon \in \mathbb{R}^{n \times d}$ to contain i.i.d. $\mathcal{N}(0,\sigma^2)$.  
Then, each entry of $\M + \epsilon$ was observed independently with probability $\p$.
Across simulations, $\n =1700$, $\d = 1000$, $r \in \{5,10,20,50\}$, $\sigma$ varies from 0.1 to 50, and $\p$ varies from 0.1 to 0.9.  
For each simulation setting, the data was sampled 100 times and the errors were averaged. 

To evaluate performance of the algorithms, we measured three different types of errors; test, training, and total errors;
the test error, $\mbox{\texttt{Test}}(\hat M) = \smallnorm{\Portho(\Mhat-\M)}_F^2/\smallnorm{\Portho(\M)}_F^2$, 
	represents the distance between the estimate $\hat{M}$ and the parameter $M_0$ measured on the unobserved entries,  
	the training error, $\mbox{\texttt{Training}}(\hat M) = \smallnorm{\P(\Mhat-\M)}_F^2/\smallnorm{\P(\M)}_F^2$, 
		the distance measured on the observed entries, and 
	the total error, $\mbox{\texttt{Total}}(\Mhat) = \smallnorm{\Mhat-\M}_F^2/\smallnorm{\M}_F^2$, 
		the distance measured on all entries.
For ease of comparison, Figure \ref{relEffvsp} and \ref{relEffvsSigma2} plot the relative efficiencies with respect to \SI-Rank.
For example, the relative test efficiency of \AIs with respect to \SI-Rank is defined as 
	$\mbox{\texttt{Test}}(\hat M_{rank}) / \mbox{\texttt{Test}}(\hat M_{adapt})$,
	where $\hat M_{adapt}$ is an estimate of \AIs and $\hat M_{rank}$ is an estimate of \emph{softImpute-Rank}. 
The relative total and training efficiencies with respect to \SI-Rank are defined similarly.

We used the best tuning parameter for the algorithms in comparison. 
Specifically, for algorithms with rank restriction (including \AI), we provided the true rank (i.e. 5, 10, 20, or 50).  
For \SI-type algorithms, an oracle tuning parameter was chosen to minimize the total error.

\begin{figure}[p]
\centering
\includegraphics[width=5.3in]{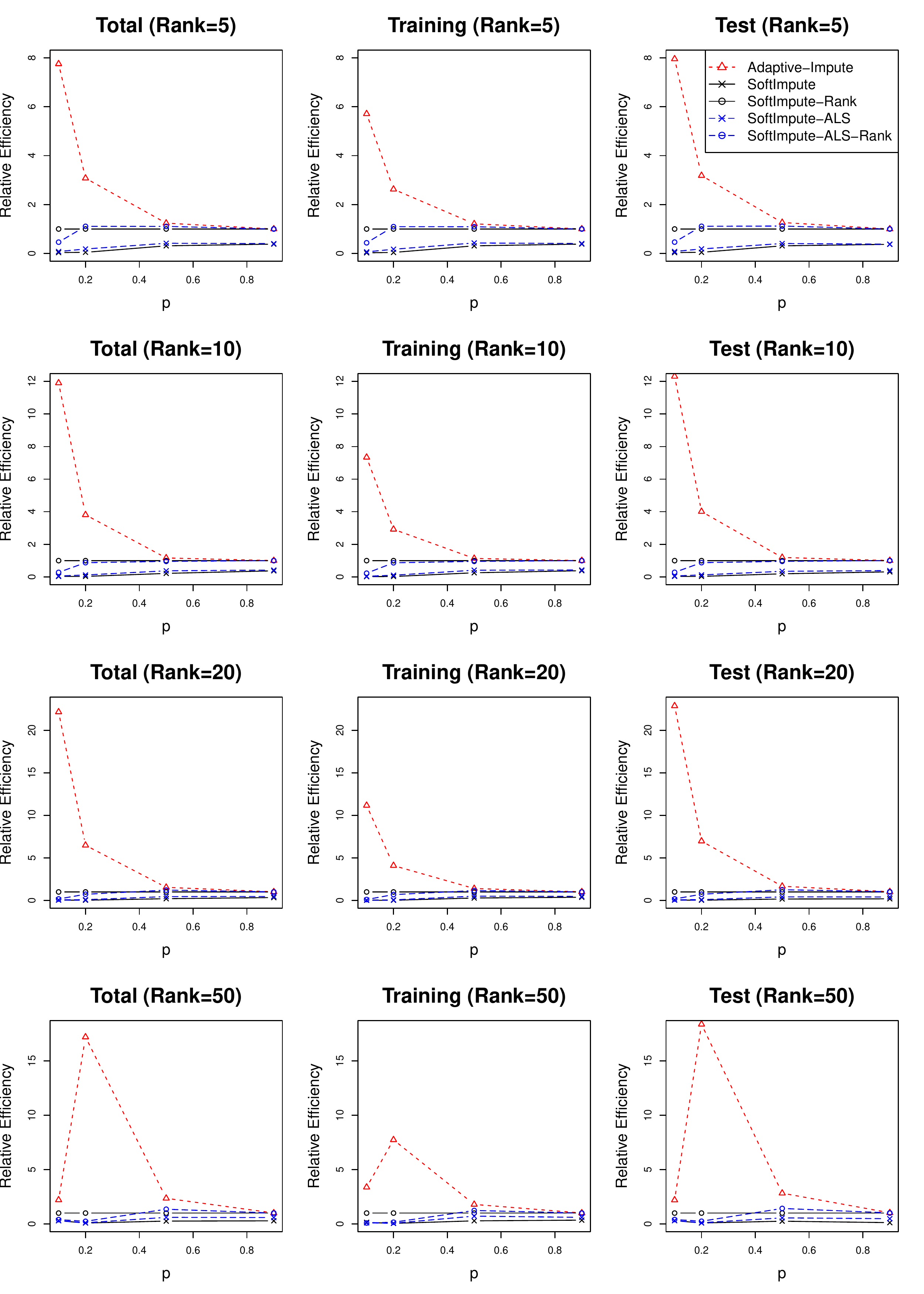}
\caption{The relative efficiency plotted against the probability of observing each entry, $\p$, when $\sig=1$. Training errors are measured over the observed entries, test errors over the unobserved entries, and total errors over all entries.}
\label{relEffvsp}
\end{figure}
\begin{figure}[p]
\centering
\includegraphics[width=5.3in]{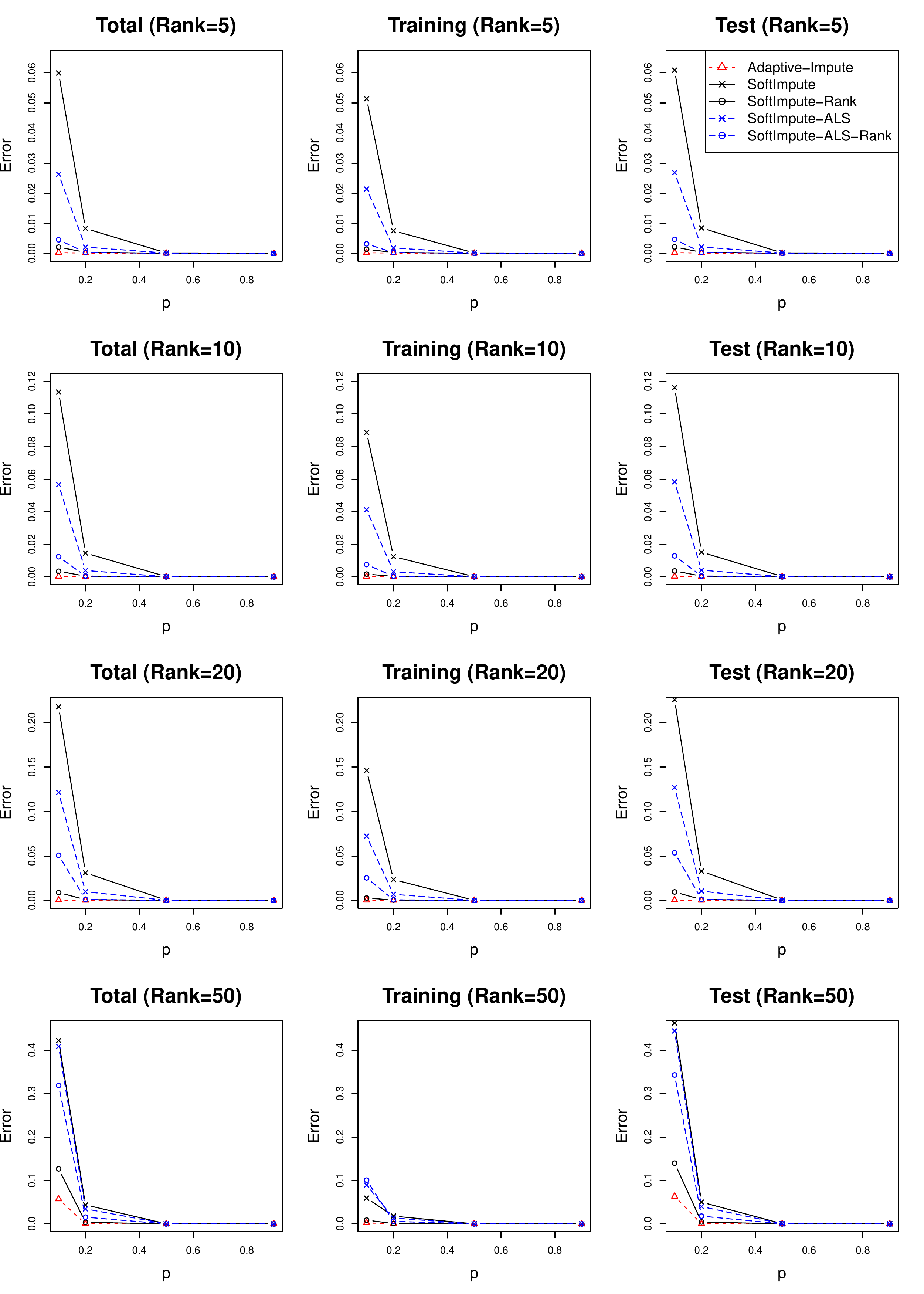}
\caption{Change of the absolute errors when the probability of observing each entry, $\p$ increases and $\sig=1$.}
\label{absEffvsp}
\end{figure}

Figure \ref{relEffvsp} shows the change of the relative efficiencies as the probability of observing each entry, $\p$, increases with $\sig=1$. 
Three columns of plots in Figure \ref{relEffvsp} correspond to three different types of errors and four rows of plots to four different values of the rank. 
In all cases, \AIs outperforms the competitors and works especially better when $\p$ is small.
Among \SI-type algorithms, the algorithms with rank constraint (i.e. \SI-Rank and \SI-ALS-Rank) perform better than the ones without (i.e. \SIs and \SI-ALS).
Figure \ref{absEffvsp} shows the change of the absolute errors that are used to compute relative efficiencies in Figure \ref{relEffvsp} as the probability of observing each entry, $\p$, increases.

Figure \ref{relEffvsSigma2} shows the change of the log relative efficiencies as the standard deviation (SD) of each entry of $\eps$, $\sigma$, increases with $\p=0.1$.
When the noise level is under 15, \AIs outperforms the competitors, but when the noise level is over 15, \SI-type algorithms start to outperform \AI. 
Hence, \SI-type algorithms are more robust to large noises than \AI. 
It may be because when there exist large noises dominating the signals, the conditions for convergence presented in Section \ref{generalSI} are not satisfied.  
In real life applications, however, it is not common to observe such large noises that dominate the signals.
Figure \ref{absEffvsSigma2} shows the change of the absolute errors that are used to compute relative efficiencies in Figure \ref{relEffvsSigma2}.

\begin{figure}[p]
\centering
\includegraphics[width=5.3in]{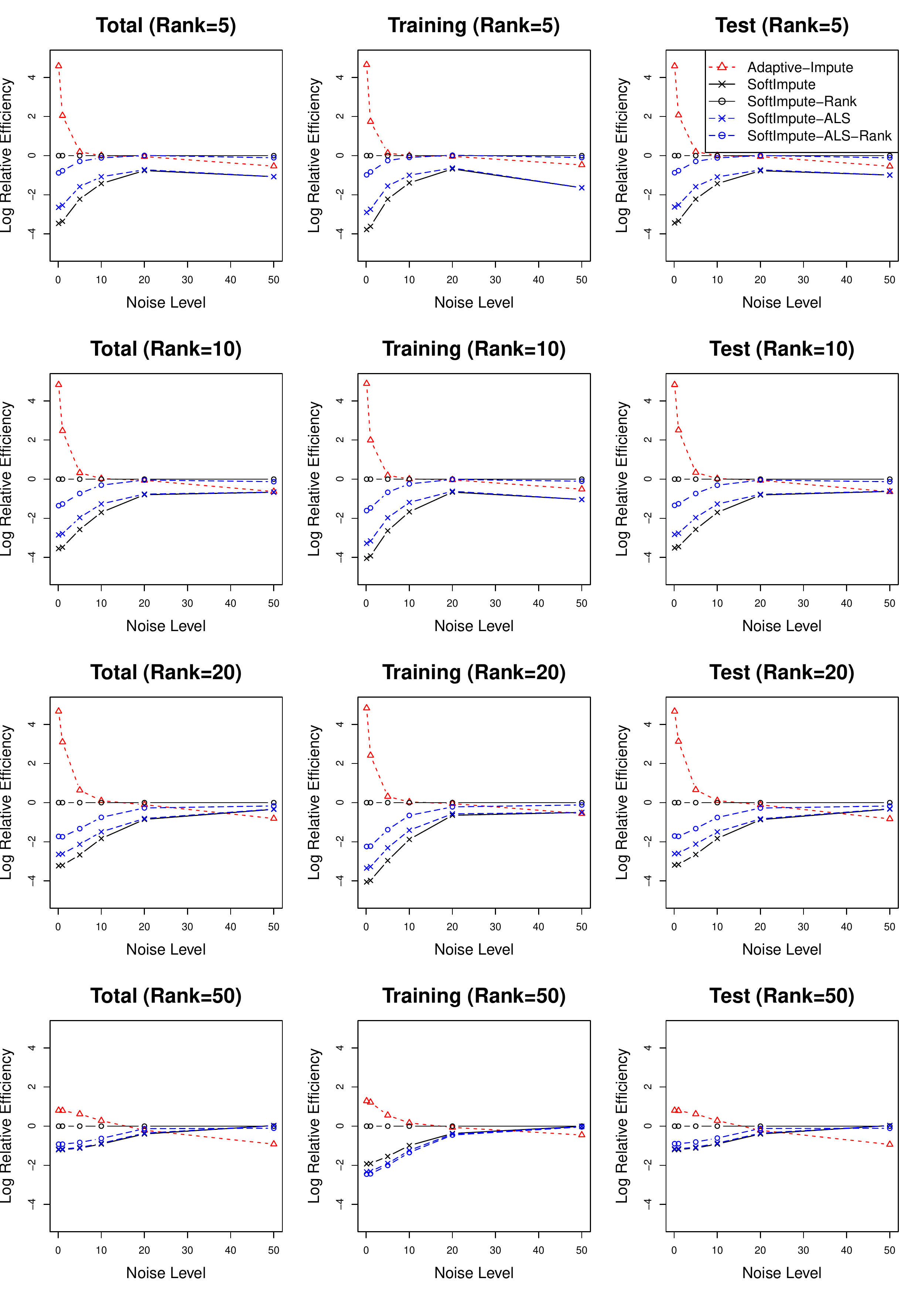}
\caption{The log relative efficiency plotted against the SD of each entry of $\eps$ when $\p=0.1$. Training errors are measured over the observed entries, test errors over the unobserved entries, and total errors over all entries.}
\label{relEffvsSigma2}
\end{figure}
\begin{figure}[p]
\centering
\includegraphics[width=5.3in]{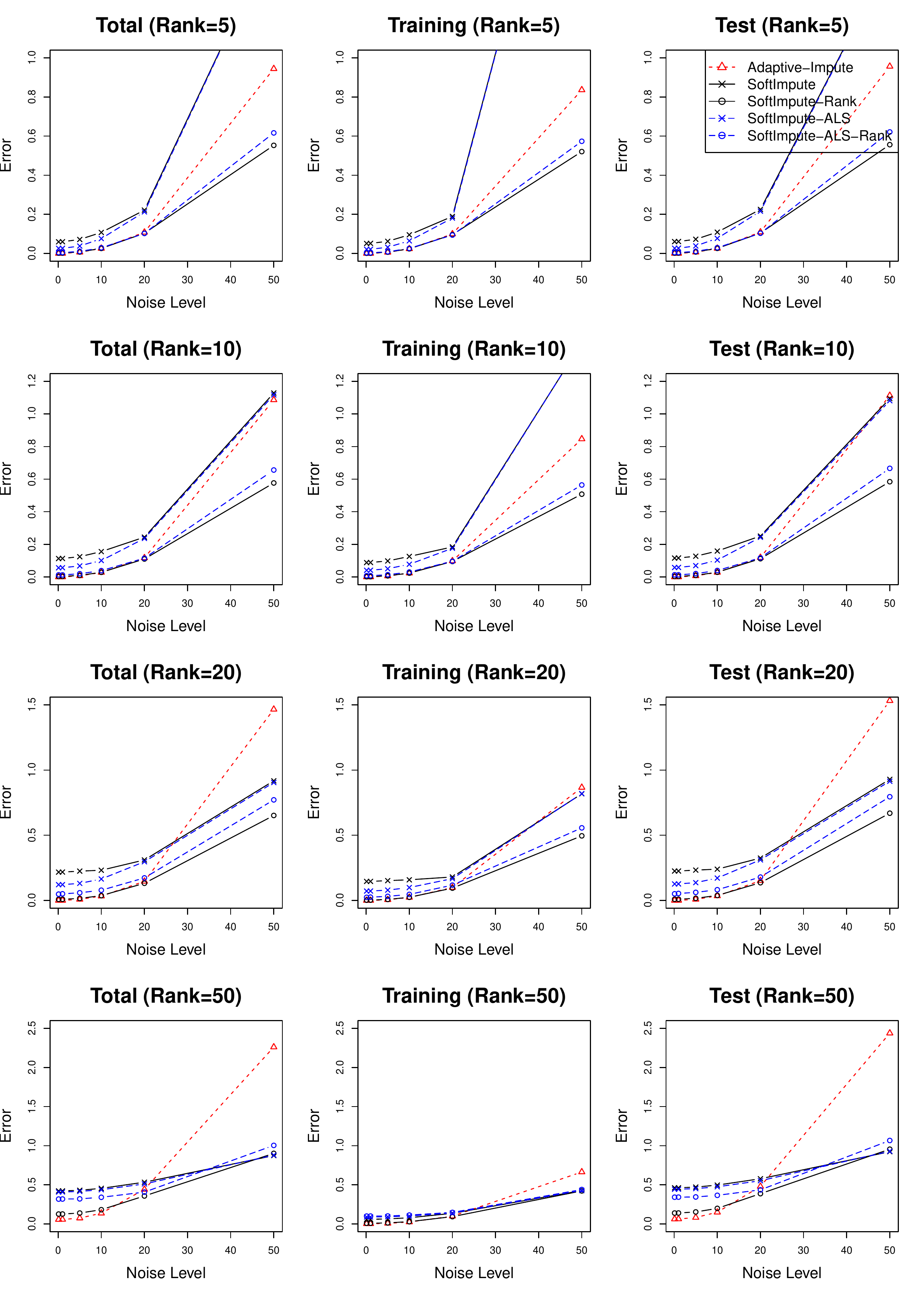}
\caption{Change of the absolute errors when the SD of each entry of $\eps$ increases and $\p=0.1$.}
\label{absEffvsSigma2}
\end{figure}

Figure \ref{convergence} shows convergence of the iterates of \AIs to the underlying low-rank matrix over iterations; 
	that is, the change of log $\mbox{\texttt{Total}}(Z_t), \mbox{\texttt{Training}}(Z_t)$, and $\mbox{\texttt{Test}}(Z_t)$ errors as $t$ increases.
Across all plots, $\n =1700$, $\d = 1000$, $\p=0.1$, and the errors were averaged over 100 replicates. 
In all cases, we observe that \AIs converges well. 
Particularly, the smaller value of noise and/or rank is, the faster \AIs converges.

\begin{figure}[p]
\centering
\includegraphics[width=5.4in]{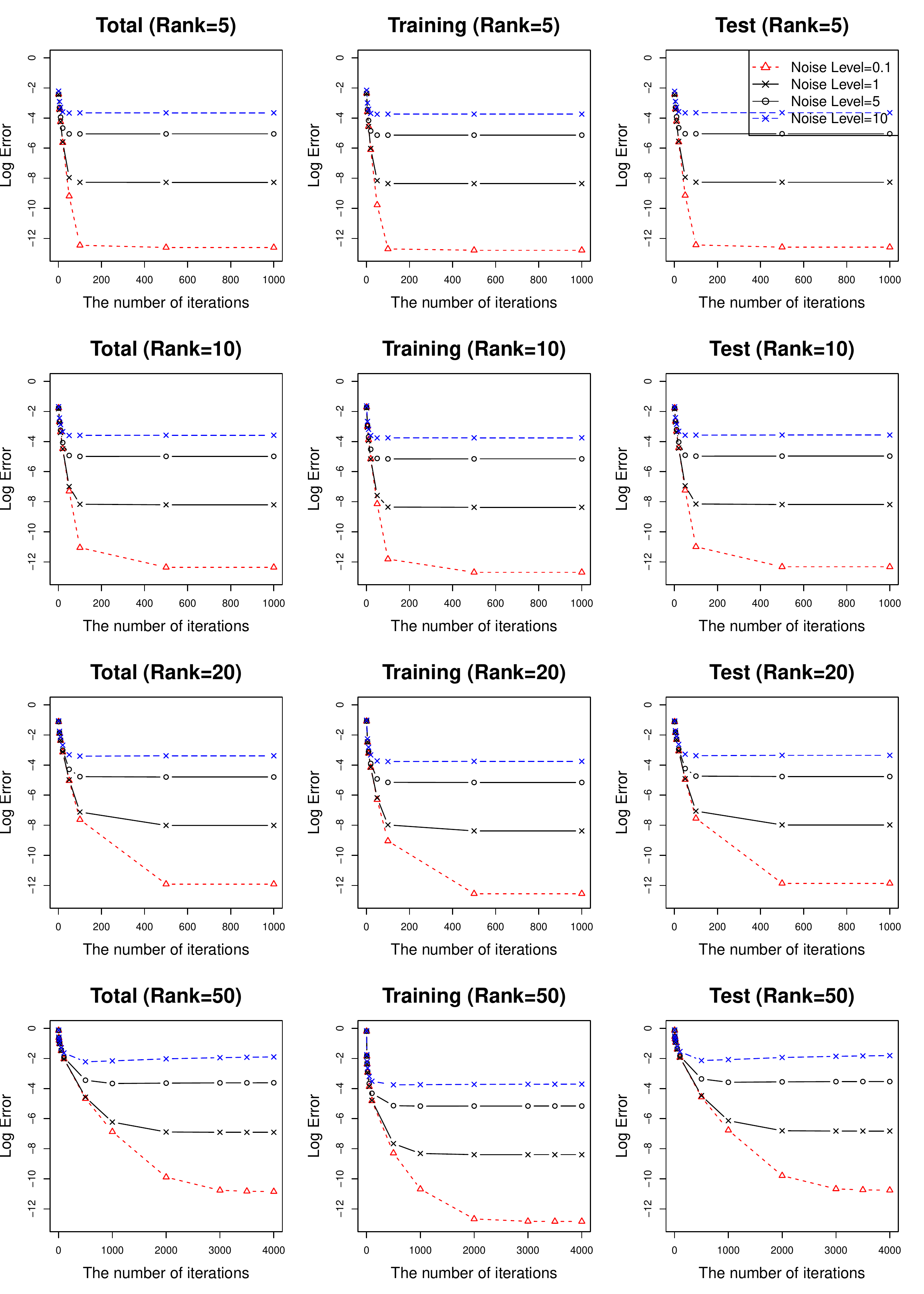}
\caption{Convergence of the iterates of \AIs to the underlying low-rank matrix. 
	In all plots, $\n =1700$, $\d = 1000$, $\p=0.1$, and all points were averaged over 100 replicates. }
\label{convergence}
\end{figure}

\subsection{A real data example} \label{realdata}

We applied \AIs and the competing methods to a real data, MovieLens 100k (\cite{movielens100k}). 
We used 5 training and 5 test data sets from 5-fold CV which are publicly available in \cite{movielens100k}.
For the rank used in \AIs and \SI-type algorithms with rank constraint, we chose 3 based on a scree plot (Figure \ref{screePlot}).
Lemma 2 in \cite{cho2015} provides justification of using the scree plot and the singular value gap to choose the rank. 
\begin{figure}[!ht]
\centering
\includegraphics[width=6in]{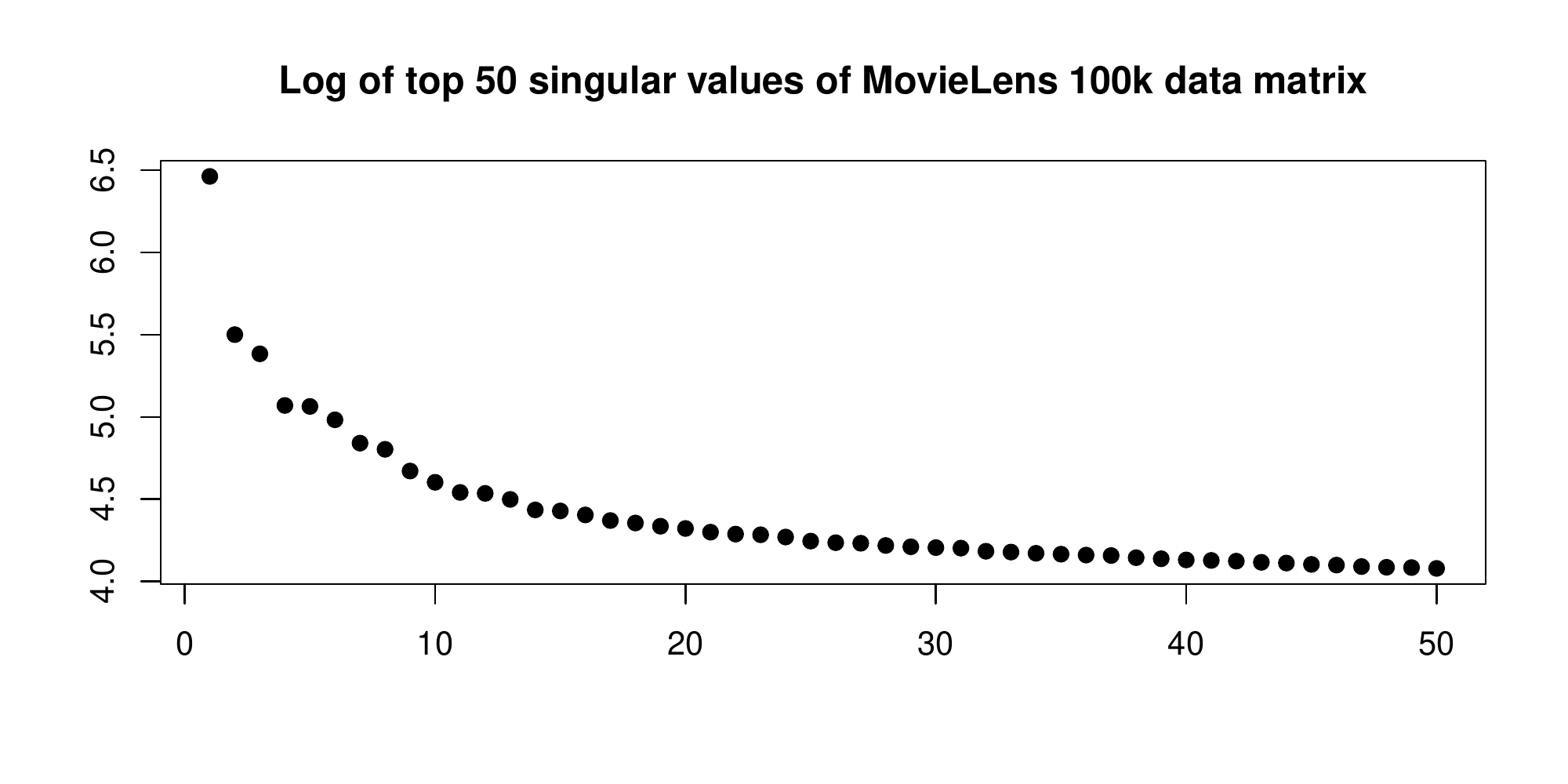}
\caption{Log of the top 50 singular values of the MovieLens 100k data matrix (\cite{movielens100k}). }
\label{screePlot}
\end{figure}
For the thresholding parameters for \SI-type algorithms, we chose the optimal values which result in the smallest test errors. 
The test errors were measured by normalized mean absolute error (NMAE) (\cite{herlocker2004}),
$$\frac{1}{(M_{max}-M_{min}) | \Omega_{test} |} \sum_{(i,j) \in \Omega_{test}} |\hat{M}_{ij} - M_{ij}|,$$
where the set $\Omega_{test}$ contains indices of the entries in test data, $|\Omega_{test}|$ is the cardinality of $\Omega_{test}$,
	$M_{max} = \max\{ \{M_{i,j}\}\setminus 0\}$ is the largest entry of $M$, and 
	$M_{min} = \min\{ \{M_{i,j}\}\setminus 0\}$ is the smallest entry of $M$.

Figure \ref{movieLens} summarizes the resulting NMAEs.
Five points in the x-axis correspond to the 5-fold CV test data, the y-axis represents the values of NMAE, and 
	the five different lines on the plane correspond to the 5 different algorithms in comparison.
We observe that \AIs outperforms all of the other algorithms. 
Specifically, the test errors of \AIs reduce those of \SI-type algorithms by 6\%-16\%.
Among \SI-type algorithms, the ones with rank constraint (i.e. \SI-Rank and \SI-ALS-Rank) performs better than the ones without (i.e. \SI and \SI-ALS).
This is the same result to the simulation results.

\begin{figure}[!ht]
\centering
\includegraphics[width=5.4in]{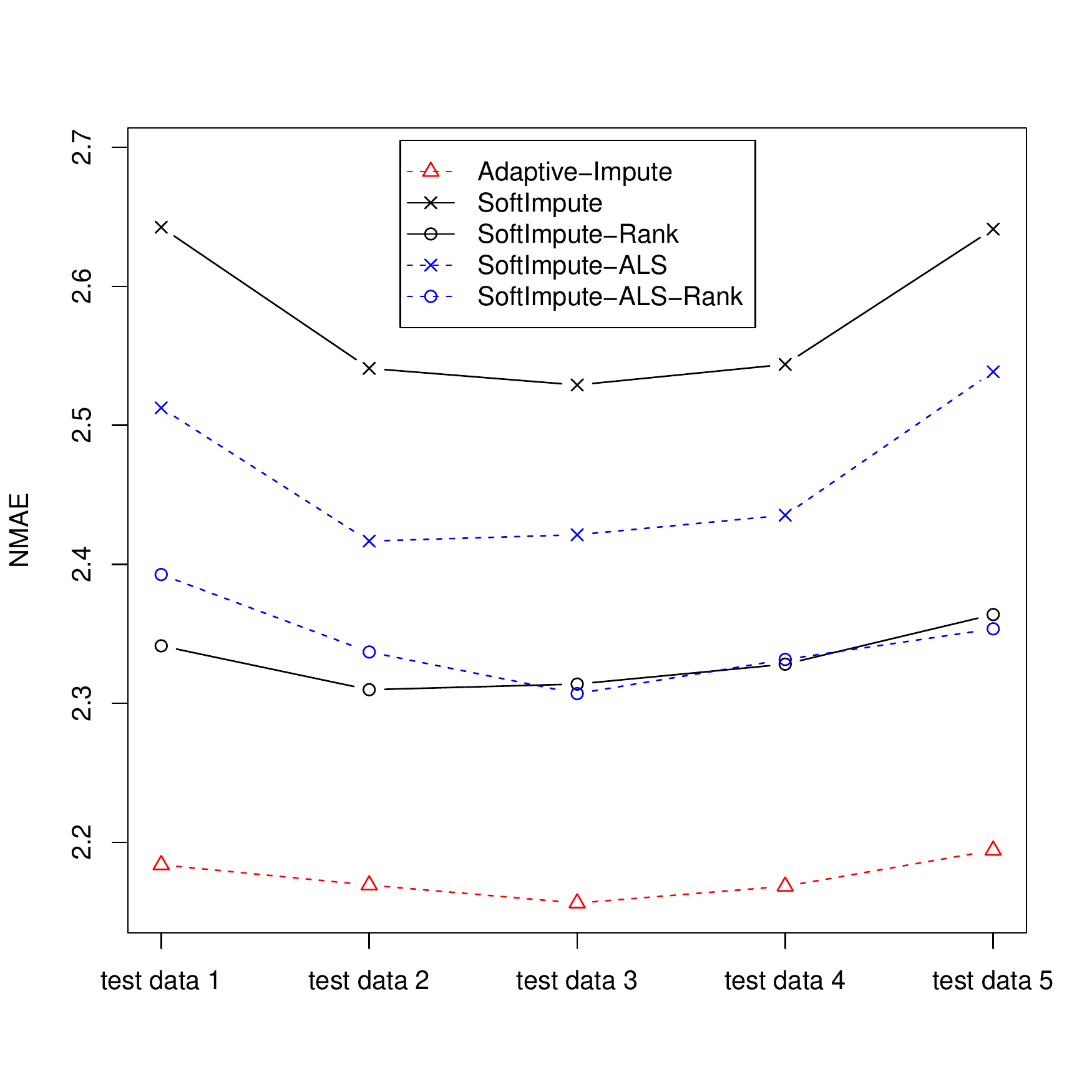}
\caption{The NMAEs of \AIs and its competitors measured in 5-fold CV test data from MovieLens 100k (\cite{movielens100k}). }
\label{movieLens}
\end{figure}



\section{Discussion} \label{discuss}
Choosing the right thresholding parameter for matrix completion algorithms using thresholded SVD often poses empirical challenges. 
This paper proposed a novel thresholded SVD algorithm for matrix completion, \AI, which employs 
	a theoretically-justified and data-dependent set of thresholding parameters.
We established its theoretical guarantees on statistical performance and showed its strong performances in both simulated and real data.
It provides understanding on the effects of thresholding and the right threshold level. 
Yet, there is a newly open problem. 
Although we proposed a reasonable remedy in the paper, the choice of the rank of the underlying low-rank matrix is of another great practical interest.
To estimate the rank and completely automate the entire procedure of \AIs would be a potential direction for future research.


\section{Proofs} \label{proofs2}

Denote by $C$ and $C_1$  generic constants whose values are free of $n$ and $p$ and may change from appearance to appearance.
Also, denote by $\smallnorm{v}_2$ the $\ell_2$-norm for any vector $v \in \real^\d$ and 
	by $\smallnorm{A}_2$ the spectral norm, the largest singular value of $A$, for any matrix $A \in \real^{\n\times\d}$.

\subsection{Proof of Theorem \ref{thm:mathInduct}}
\begin{proof}[Proof of Theorem \ref{thm:mathInduct}]
We have
\begin{eqnarray*}
\Mtilde_t 
&=& \P(\Mp) + \Portho(Z_t) 
\cr
&=& \y\cdot(\M + \eps) + (1_{\n}1_{\d}^T-\y)\cdot Z_t
\cr
&=& \M + \y\cdot\eps + (1_{\n}1_{\d}^T-\y)\cdot\eta_t,
\end{eqnarray*}
where $1_{\n}$ and $1_{\d}$ are vectors of length $\n$ and $\d$, respectively, filled with ones and 
	$\eta_t = Z_t -\M$, and, $A \cdot B= (A_{ij} B_{ij}) _{1\leq i \leq n, 1 \leq j \leq d}$ for any $A$ and $B \in \mathbb{R}^{n \times d}$.
Assume that 
\begin{eqnarray} \label{init-assume}
\frac{1}{\sqrt{\n\d}}\norm{\eta_t}_F = o_p\(\sqrt{\frac{h_\n}{\p\,\d}}\).
\end{eqnarray}
 Then, simple algebraic manipulations show for large $\n$
\begin{eqnarray} \label{eta_t+1}
\frac{1}{\sqrt{\n\d}}\norm{\eta_{t+1}}_F
&=&\frac{1}{\sqrt{\n\d}}\norm{Z_{t+1}-\M}_F
\cr
&=&\frac{1}{\sqrt{\n\d}}\norm{\sum_{i=1}^\rank \sqrt{\blam_i^2(\Mtilde_t) - \alphatilde_{t}} \, \bu_i(\Mtilde_t) \bv_i(\Mtilde_t)^T 
	- \sum_{i=1}^\rank \lam_i\U_i\V_i^T}_F
\cr
&\leq& \frac{C}{\sqrt{\n\d}}\sum_{i=1}^\rank \Bigg\{ \left| \sqrt{\blam_i^2(\Mtilde_t) - \alphatilde_{t}} - \lam_i \right|   \norm{\U_i\V_i^T}_F 
\cr
&&\quad\quad\quad\quad
	+ \lam_i \norm{\(\bu_i(\Mtilde_t)-\U_i\O_i\) \V_i^T }_F 
	+ \lam_i \norm{ \U_i \(\bv_i(\Mtilde_t)-\V_i\Q_i\)^T }_F \Bigg\}
\cr
&\le&C \sum_{i=1}^\rank \Bigg\{ \frac{1}{\sqrt{\n\d}}\left| \sqrt{\blam_i^2(\Mtilde_t) - \alphatilde_{t}} - \lam_i \right| 
\cr
&&\quad\quad\quad\quad
	+ \frac{\lam_i}{\sqrt{\n\d}} \norm{\bu_i(\Mtilde_t)-\U_i\O_i}_F 
	+ \frac{\lam_i}{\sqrt{\n\d}} \norm{\bv_i(\Mtilde_t)-\V_i\Q_i}_F \Bigg\},
\end{eqnarray}
where $\O_i$ and $\Q_i$ are in $\{-1,1\}$ and minimize $\norm{\bu_i(\Mtilde_t)-\U_i\O_i }_F$ and $\norm{ \bv_i(\Mtilde_t)-\V_i\Q_i }_F$, respectively.

To find the order of \eqref{eta_t+1}, first consider the term $\norm{\bv_i(\Mtilde_t)-\V_i\O_i }_F$.
By Davis-Kahan Theorem (Theorem 3.1 in \cite{li1998two}) and Proposition 2.2 in \cite{vu2013}, 
\begin{eqnarray} \label{utilde-u}
\norm{\bv_i(\Mtilde_t)-\V_i\O_i }_F \le \frac{\frac{1}{\n\d}\norm{\(\Mtilde_t^T\Mtilde_t - \[ \M^T\M +\n \p \sig^2 I \] \)\V_i}_F }{\left| \frac{1}{\n\d}\( \lam_i^2+\n \p \sig^2 - \blam_{i+1}^2(\Mtilde_t) \) \right| }.
\end{eqnarray}
Consider the numerator of \eqref{utilde-u}.
We have 
\begin{eqnarray} \label{vtilde-v}
&&\frac{1}{\n\d}\norm{\(\Mtilde_t^T\Mtilde_t - \[ \M^T\M + \n \p \sig I\]\)\V_i}_F
\cr
&&\le \frac{1}{\n\d}\bigg\{
	\norm{\(\y\cdot\eps\)^T\(\y\cdot\eps\)\V_i- \n \p \sig^2\V_i}_F + \norm{\[(1_{\n}1_{\d}^T-\y)\cdot\eta_t\]^T\[(1_{\n}1_{\d}^T-\y)\cdot\eta_t\]\V_i}_F
\cr
&&\quad\quad\quad\quad
	+ \norm{\M^T\(\y\cdot\eps\)\V_i}_F + \norm{\(\y\cdot\eps\)^T\M\V_i}_F
\cr
&&\quad\quad\quad\quad
	+ \norm{\M^T\[(1_{\n}1_{\d}^T-\y)\cdot\eta_t\]\V_i}_F + \norm{\[(1_{\n}1_{\d}^T-\y)\cdot\eta_t\]^T\M\V_i}_F
\cr
&&\quad\quad\quad\quad
	+ \norm{\[(1_{\n}1_{\d}^T-\y)\cdot\eta_t\]^T\(\y\cdot\eps\)\V_i}_F + \norm{\(\y\cdot\eps\)^T\[(1_{\n}1_{\d}^T-\y)\cdot\eta_t\]\V_i}_F \bigg\}
\cr
&&= \frac{1}{\n\d} \Big\{ O_p\(\p\sqrt{\n\d}\) + o_p\(\frac{\n h_\n}{\p}\) + O_p\(\sqrt{\p\d\n^2}\) + O_p\(\sqrt{\p\d^2\n}\)
\cr
&&\quad\quad\quad\quad
	+ o_p\(\sqrt{\frac{h_\n \d\n^2}{\p}}\) + o_p\(\sqrt{\frac{h_\n \d\n^2}{\p}}\) + o_p\(\sqrt{h_\n \n^2}\) + o_p\(\sqrt{h_\n \d\n^2}\) \Big\}
\cr
&&=o_p\(\sqrt{\frac{h_\n}{\p\d}}\),
\end{eqnarray}
where the first equality holds due to \eqref{mod}, Assumption \ref{assume1}(2), \eqref{init-assume}, and \eqref{epsyTepsy} and \eqref{yepsu-yepsv} below.
We have 
\begin{eqnarray} \label{epsyTepsy}
&&\expect \norm{\(\y\cdot\eps\)^T\(\y\cdot\eps\)\V_i- \n \p \sig^2\V_i}_F^2
\cr
&&= \expect\left\{\sum_{h=1}^\d \[ \sum_{k=1}^\n \sum_{j=1}^\d \( \y_{kh}\y_{kj}\eps_{kh}\eps_{kj}\V_{ij} -\p\sig^2\V_{ih}\I_{(j=h)} \) \]^2\right\}
\cr
&&= \sum_{h=1}^\d \sum_{k=1}^\n \sum_{j=1}^\d \expect\( \y_{kh}\y_{kj}\eps_{kh}\eps_{kj}\V_{ij} -\p\sig^2\V_{ih}\I_{(j=h)} \)^2
\cr
&&= \sum_{k=1}^\n \sum_{j\ne h}^\d \V_{ij}^2 \,\expect\( \y_{kh}^2\y_{kj}^2\eps_{kh}^2\eps_{kj}^2\) 
		+ \sum_{k=1}^\n \sum_{j=h}^\d \V_{ij}^2 \,\expect\( \y_{kj}^2\eps_{kj}^2 -\p\sig^2 \)^2
\cr
&&= O\(\p^2\n\d\),
\end{eqnarray}
where $\V_{ij}$ is the $j$-th element of $\V_i$.
Similarly, we have
\begin{eqnarray} \label{yepsu-yepsv}
\expect\smallnorm{\(\y\cdot\eps\)\V_i}_F^2=O\(\p\n\),
\expect\smallnorm{\U_i^T\(\y\cdot\eps\)}_F^2=O\(\p\d\),
\text{ and }\,
\expect\smallnorm{\y\cdot\eps}_F^2=O\(\p\n\d\).
\end{eqnarray}

Consider the denominator of \eqref{utilde-u}.
 By Weyl's theorem (Theorem 4.3 in \cite{li1998}), we have
 \begin{eqnarray} \label{deltaExist}
&&\max_{1\le i\le\d}\frac{1}{\n\d}|\lam_i^2  +\n \p \sig^2 -  \blam_{i}^2(\Mtilde_t)| 
\cr
&&\le \frac{1}{\n\d}\norm{\Mtilde_t^T\Mtilde_t -  \[ \M^T\M + \n \p \sig ^2 I \]}_2
\cr
&&\le \frac{1}{\n\d}\bigg\{ \norm{\(\y\cdot\eps\)^T\(\y\cdot\eps\)-  \n \p \sig ^2 I }_2 + \norm{\[(1_{\n}1_{\d}^T-\y)\cdot\eta_t\]^T\[(1_{\n}1_{\d}^T-\y)\cdot\eta_t\]}_2 
\cr
&&\quad\quad\quad\quad
	+ 2\norm{ \M^T\(\y\cdot\eps\)}_2 + 2 \norm{\M^T\[(1-\y)\cdot\eta_t\]}_2 + 2\norm{\(\y\cdot\eps\)^T\[(1_{\n}1_{\d}^T-\y)\cdot\eta_t\]}_2 \bigg\}
\cr
&&= \frac{1}{\n\d}\bigg\{ O_p\(\p\sqrt{\n\d^2}\) + o_p\(\frac{\n h_\n}{\p}\) + O_p\(\sqrt{\p\,\d\n^2}\) + o_p\(\sqrt{\frac{\d\n^2h_\n}{\p}}\) + o_p\(\sqrt{\d\n^2h_\n}\) \bigg\}
\cr
&&=o_p(1),
\end{eqnarray}
where the last two lines holds similarly to \eqref{vtilde-v}.

Thus, by  \eqref{deltaExist} and \eqref{vtilde-v},
\begin{eqnarray} \label{v-v}
\smallnorm{\bv_i(\Mtilde_t)-\V_i\O_i }_F = o_p\(\sqrt{\frac{h_\n}{\p\d}}\). 
\end{eqnarray}

Secondly, similar to the proof of \eqref{v-v},
we can show  $\smallnorm{\bu_i(\Mtilde_t)-\U_i\O_i }_F = o_p\(\sqrt{h_\n/\p\d}\)$.

Lastly, consider the term $\frac{1}{\sqrt{\n\d}}\left| \sqrt{\blam_i^2(\Mtilde_t) - \alphatilde_{t}} - \lam_i \right|$. 
By Taylor's expansion, there is $\lam_\ast^2$  between $\blam_i^2(\Mtilde_t) - \alphatilde_{t}$ and $\lam_i^2$ such that 
\begin{eqnarray} \label{lam-lam}
&&\frac{1}{\sqrt{\n\d}}\left| \sqrt{\blam_i^2(\Mtilde_t) - \alphatilde_{t}} - \lam_i \right| 
\cr
&&= \frac{1}{\sqrt{\n\d}}\left| \frac{1}{2\lam_\ast}\( \blam_i^2(\Mtilde_t) - \alphatilde_{t} - \lam_i^2 \) \right|
\cr
&&\le \frac{1}{2\lam_\ast\sqrt{\n\d}}\left| \blam_i^2(\Mtilde_t) - (\lam_i^2+\n\p\sig^2) \right| 
	+ \frac{1}{2\lam_\ast\sqrt{\n\d}}\left| \alphatilde_{t} - \n\p\sig^2 \right|.
\end{eqnarray}
We need to find the convergence rates of $\frac{1}{\n\d}\left| \blam_i^2(\Mtilde_t) - (\lam_i^2+\n\p\sig^2) \right|$ and $\frac{1}{\n\d}\left| \alphatilde_{t} - \n\p\sig^2 \right|$.
Let $\V_c = \(\V_{\rank+1},\ldots,\V_\d\) \in \real^{\d\times(\d-\rank)}$ be a matrix 
	such that $\V_c^T\V_c = I_{\d-\rank}$ and $\V^T\V_c = 0_{\rank\times(\d-\rank)}$
and let $\widetilde\V_{t} = \(\bv_{1}(\Mtilde_t), \ldots, \bv_\d(\Mtilde_t)\) \in \real^{\d\times\rank}$ and 
	$\widetilde\V_{tc} = \(\bv_{\rank+1}(\Mtilde_t), \ldots, \bv_\d(\Mtilde_t)\) \in \real^{\d\times(\d-\rank)}$
	so that $\widetilde\V_{t}^T\widetilde\V_{tc}=0_{\rank\times(\d-\rank)}$.
Also, let $\O = \diag(\O_1, \ldots, \O_\rank)$ and $\O_c = \diag(\O_{\rank+1}, \ldots, \O_\d)$, where 
$$\O_i := \argmin_{o \in \{-1,1\}} \norm{\V_i\,o-\bv_i(\Mtilde_t)}_2^2 \quad\text{for }\; i=1,\ldots,\d.$$
Then, we have
\begin{eqnarray} \label{lam2-lam2}
&&\frac{1}{\n\d}\left| \blam_i^2(\Mtilde_t) - (\lam_i^2+\n\p\sig^2) \right|
\cr
&&= \frac{1}{\n\d}\bigg| \bv_i(\Mtilde_t)^T\Mtilde_t^T\Mtilde_t \bv_i(\Mtilde_t)  
	- \V_i^T\( {\M}^T{\M} + \n \p \sig^2 I \) \V_i \bigg|
\cr
&&\le \frac{1}{\n\d}\bigg| \V_i^T \[\Mtilde_t^T\Mtilde_t - \( {\M}^T{\M}+ \n \p \sig^2 I \) \] \V_i   \bigg|
	+ \frac{1}{\n\d}\bigg| \bv_i(\Mtilde_t)^T\Mtilde_t^T\Mtilde_t \bv_i(\Mtilde_t) - \V_i^T\Mtilde_t^T\Mtilde_t\V_i \bigg| 
\cr
&&\le o_p\( \sqrt{\frac{h_n}{\p\d}} \) + \frac{1}{\n\d}\bigg| \bv_i(\Mtilde_t)^T\Mtilde_t^T\Mtilde_t \bv_i(\Mtilde_t) - \V_i^T\Mtilde_t^T\Mtilde_t\V_i \bigg| 
\cr
&&= o_p\(\sqrt{\frac{h_\n}{\p\d}}\),
\end{eqnarray}
where the second  inequality can be derived by the similar way to the proof of \eqref{vtilde-v}, 
and  the last equality is due to \eqref{eq 11 : Thm asym} below.
Simple algebraic manipulations show
\begin{eqnarray} \label{eq 11 : Thm asym}
&&\frac{1}{\n\d}\bigg| \bv_i(\Mtilde_t)^T\Mtilde_t^T\Mtilde_t \bv_i(\Mtilde_t) - \V_i^T\Mtilde_t^T\Mtilde_t\V_i \bigg|  \cr
&&=   \frac{1}{\n\d}\bigg| \[\V_i\O_i-\bv_i(\Mtilde_t)\]^T\Mtilde_t^T\Mtilde_t \[\V_i\O_i-\bv_i(\Mtilde_t)\] + 2\blam_i^2(\Mtilde_t)\[\V_i\O_i-\bv_i(\Mtilde_t)\]^T\bv_i(\Mtilde_t) \bigg|
\cr
&&\le  \frac{2\blam_1^2(\Mtilde_t)}{\n\d}\norm{\V_i\O_i-\bv_i(\Mtilde_t)}_2^2
\cr
&&=   o_p\(\frac{ h_\n}{\p\d }\),
\end{eqnarray}
where the last equality is due to \eqref{deltaExist} and \eqref{v-v}.
Also,
\begin{eqnarray} \label{rhohat-npsig2}
&&\frac{1}{\n\d}\left| \alphatilde_{t} - \n\p\sig^2 \right|
\cr
&&= \frac{1}{\n\d}\bigg| \frac{1}{\d-\rank}\sum_{j=\rank+1}^\d \bv_j(\Mtilde_t)^T\Mtilde_t^T\Mtilde_t\bv_j(\Mtilde_t) - \n\p\sig^2 \bigg| 
\cr
&&\le \frac{1}{\n\d}\bigg| \frac{1}{\d-\rank}\sum_{j=\rank+1}^\d \V_j^T\[ \Mtilde_t^T\Mtilde_t - \( \M^T\M + \n \p \sig ^2 I\)\] \V_j \bigg| 
\cr
&&\quad\quad\quad\quad
	+ \frac{1}{\n\d}\bigg| \frac{1}{\d-\rank}\sum_{j=\rank+1}^\d \[ \bv_j(\Mtilde_t)^T \Mtilde_t^T\Mtilde_t \bv_j(\Mtilde_t) - \V_j^T \Mtilde_t^T\Mtilde_t \V_j \] \bigg| \cr
&&=   o_p\(\sqrt{\frac{h_\n}{\p\d} }\)+ \frac{1}{\n\d}\bigg| \frac{1}{\d-\rank}\sum_{j=\rank+1}^\d \[ \bv_j(\Mtilde_t)^T \Mtilde_t^T\Mtilde_t \bv_j(\Mtilde_t) - \V_j^T \Mtilde_t^T\Mtilde_t \V_j \] \bigg| \cr 
&&=  o_p\(\sqrt{ \frac{h_\n}{\p\d} }\),
\end{eqnarray}
where the second equality can be derived by the similar way to the proof of \eqref{vtilde-v}, and the last equality is due to  \eqref{eq 12 : Thm asym} below.
Similar to the proof of \eqref{eq 11 : Thm asym}, we have
\begin{eqnarray}\label{eq 12 : Thm asym}
&&  \frac{1}{\n\d(\d-\rank)}\sum_{j=\rank+1}^\d \bigg| \bv_j(\Mtilde_t)^T \Mtilde_t^T\Mtilde_t \bv_j(\Mtilde_t) - \V_j^T \Mtilde_t^T\Mtilde_t \V_j \bigg| 
\cr
&&\le  \frac{1}{\n\d(\d-\rank)}\sum_{j=\rank+1}^\d 2\blam_1^2(\Mtilde_t)\norm{\V_j\O_j-\bv_j(\Mtilde_t)}_2^2
\cr
&&\le   \frac{2\blam_1^2(\Mtilde_t)}{\n\d(\d-\rank)}\norm{\V_c\O_c-\widetilde\V_{tc}}_F^2
\cr
&&\le  \frac{4\blam_1^2(\Mtilde_t)}{\n\d(\d-\rank)}\norm{\V_c\V_c^T - \widetilde\V_{tc}\widetilde\V_{tc}^T}_F^2
\cr
&&=   \frac{4\blam_1^2(\Mtilde_t)}{\n\d(\d-\rank)}\norm{\V\V^T - \widetilde\V_{t}\widetilde\V_{t}^T}_F^2
\cr
&&\le   \frac{4\blam_1^2(\Mtilde_t)}{\n\d(\d-\rank)}\norm{\V\O-\widetilde\V_{t}}_F^2
\cr
&&=   \frac{4\blam_1^2(\Mtilde_t)}{\n\d(\d-\rank)}\sum_{i=1}^\rank\norm{\V_i\O_i-\bv_i(\Mtilde_t)}_2^2
\cr
&&=  o_p\(\frac{h_\n}{\p\d^2}\),
\end{eqnarray}
where the  fourth and sixth lines are  due to Proposition 2.2 in \cite{vu2013}, and the last line holds from \eqref{v-v}.

The three results above \eqref{lam-lam}, \eqref{lam2-lam2}, and \eqref{rhohat-npsig2} give $\frac{1}{\sqrt{\n\d}}\left| \sqrt{\blam_i^2(\Mtilde_t) - \alphatilde_{t}} - \lam_i \right|=o_p\(\sqrt{\frac{h_\n}{\p\d}}\)$.

Therefore, combining the results above, we have that $\frac{1}{\sqrt{\n\d}}\norm{\eta_{t+1}}_F$ in \eqref{eta_t+1} is $o_p\(\sqrt{\frac{h_\n}{\p\d}}\)$.
Since $\frac{1}{\sqrt{\n\d}}\norm{\eta_{1}}_F=o_p\(\sqrt{\frac{h_\n}{\p\d}}\)$ by Proposition \ref{MhatConsistent}, 
	we have $\frac{1}{\sqrt{\n\d}}\norm{\eta_{t}}_F=o_p\(\sqrt{\frac{h_\n}{\p\d}}\)$ for any fixed $t$ by mathematical induction.
\end{proof}

\subsection{Proof of Theorem \ref{solutionZ}}
\begin{proof}[Proof of Theorem \ref{solutionZ}]
We have 
\begin{eqnarray} \label{targetOptim}
&&\min_Z \;\frac{1}{2\n\d}\norm{X-Z}_F^2 + \sum_{i=1}^\d \frac{\tau_i}{\sqrt{\n\d}} \frac{\blam_i(Z)}{\sqrt{\n\d}}
\cr
&&= \min_Z \;
	\frac{1}{2\n\d} \Bigg\{ \norm{X}_F^2 -2\sum_{i=1}^\d \tilde\lam_i\cdot\tilde u_i^TX\tilde v_i + \sum_{i=1}^\d\tilde\lam_i^2\Bigg\}
		+ \frac{1}{\n\d} \sum_{i=1}^\d \tau_i\tilde\lam_i,
\end{eqnarray}
where $\tilde\lam_i=\blam_i(Z)$, $\tilde u_i=\bu_i(Z)$, and $\tilde v_i=\bv_i(Z)$. 
Minimizing \eqref{targetOptim} is equivalent to minimizing
\begin{eqnarray*}
-2\sum_{i=1}^\d \tilde\lam_i\cdot\tilde u_i^TX\tilde v_i + \sum_{i=1}^\d\tilde\lam_i^2
		+ \sum_{i=1}^\d 2\tau_i\tilde\lam_i,
\end{eqnarray*}
with respect to $\tilde\lam_i,\tilde u_i,$ and $\tilde v_i, \; i=1,\ldots,\d$,
under the conditions that $(\tilde u_1,\ldots,\tilde u_\d)^T(\tilde u_1,\ldots,\tilde u_\d)=I_\d$,
	$(\tilde v_1,\ldots,\tilde v_\d)^T(\tilde v_1,\ldots,\tilde v_\d)=I_\d$, and $\tilde\lam_1 \ge \tilde\lam_2 \geq \ldots \ge \tilde\lam_d \ge 0$.
Thus, we have
\begin{eqnarray} \label{targetOptim-final}
&&\min_{\tilde\lam_i\ge0,\tilde u_i,\tilde v_i, \; i=1,\ldots,\d} \;
	-2\sum_{i=1}^\d \tilde\lam_i\cdot\tilde u_i^TX\tilde v_i + \sum_{i=1}^\d\tilde\lam_i^2
		+ \sum_{i=1}^\d 2\tau_i\tilde\lam_i
\cr
&&= \min_{\tilde\lam_i\ge0, \; i=1,\ldots,\d} \;
	-2\sum_{i=1}^\d \tilde\lam_i\cdot\blam_i(X) + \sum_{i=1}^\d\tilde\lam_i^2
		+ \sum_{i=1}^\d 2\tau_i\tilde\lam_i
\cr
&&= \min_{\tilde\lam_i\ge0, \; i=1,\ldots,\d} \;
	\sum_{i=1}^\d \left\{ \tilde\lam_i^2 -2 \tilde\lam_i \[\blam_i(X) - \tau_i\] \right\},
\end{eqnarray} 
where the first equality is due to the facts that $\tilde\lam_1\ge\ldots\ge\tilde\lam_\d\ge0$, and  for every $i$,  the problem
$$\max_{\norm{u_i}_2^2\le1, \norm{v_i}_2^2\le1} u_i ^TXv_i 
	\quad\text{such that}\quad 
	u_i \perp\{\tilde u_1^\ast,\ldots,\tilde u_{i-1}^\ast\}, v_i \perp\{\tilde v_1^\ast,\ldots,\tilde v_{i-1}^\ast\}$$ 
is solved by $\tilde u_i^\ast, \tilde v_i^\ast$, the left and right singular vectors of $X$ corresponding to the $i$-th largest singular value of $X$. 
Note that $\tilde u_i = \tilde u_i ^\ast$.
Since \eqref{targetOptim-final} is a quadratic function of $\tilde\lam_i$, the solution to the problem \eqref{targetOptim-final} is then $\tilde\lam_i = \(\blam_i(X) - \tau_i\)_+$.
\end{proof}


\subsection{Proof of Theorem \ref{thm:Zsolvesf}}
To ease the notation, we drop the superscript `g' in $Z_t^g$, $\Mtilde_t^g$, and $D_t^g$ in this section.

\begin{lemma}\label{z-zgoto0}
Let $Z_{t+1}:=\argmin_{Z\in\real^{\n\times\d}} Q_\tau(Z|Z_t)$ in \eqref{eq:Q}. Then, under Assumption \ref{assume2}, we have
$$\norm{Z_{t+1}-Z_{t}}_F^2 \to 0 \quad\text{as } t\to\infty.$$
\end{lemma}
\begin{proof}[Proof of Lemma \ref{z-zgoto0}]
By the construction of $D_t$, 
$$(\Mtilde_{t-1}-\Mtilde_t) - (Z_t - Z_{t+1}) - (D_{t-1} - D_{t}) =0.$$
Thus, we have 
\begin{eqnarray} \label{xZ-Z}
\langle\Mtilde_{t-1}-\Mtilde_t,Z_t - Z_{t+1}\rangle - \langle Z_t - Z_{t+1},Z_t - Z_{t+1}\rangle - \langle D_{t-1} - D_{t},Z_t - Z_{t+1}\rangle =0
\end{eqnarray}
and
\begin{eqnarray}\label{xM-M}
\langle\Mtilde_{t-1}-\Mtilde_t,\Mtilde_{t-1}-\Mtilde_t\rangle - \langle Z_t - Z_{t+1},\Mtilde_{t-1}-\Mtilde_t\rangle - \langle D_{t-1} - D_{t},\Mtilde_{t-1}-\Mtilde_t\rangle =0.
\end{eqnarray}
Add \eqref{xM-M} and \eqref{xZ-Z}, and 
\begin{eqnarray}\label{xZ-Z:xM-M}
&&0 =\smallnorm{\Mtilde_{t-1}-\Mtilde_t}_F^2 - \smallnorm{Z_t-Z_{t+1}}_F^2 - \langle D_{t-1} - D_{t}, Z_t + \Mtilde_{t-1} - (Z_{t+1}+\Mtilde_t)\rangle
\cr
&&=\smallnorm{\Mtilde_{t-1}-\Mtilde_t}_F^2 - \smallnorm{Z_t-Z_{t+1}}_F^2 - \norm{D_{t-1} - D_{t}}_F^2 - 2\langle D_{t-1} - D_{t},Z_t - Z_{t+1} \rangle.
\end{eqnarray}
Under Assumption \ref{assume2}, \eqref{xZ-Z:xM-M} gives  
\begin{eqnarray*}
\norm{Z_t-Z_{t+1}}_F^2 \le \norm{\Mtilde_{t-1} - \Mtilde_t}_F^2,
\end{eqnarray*}
and thus
\begin{eqnarray} \label{lem:z-z<z-z}
\norm{Z_{t+1}-Z_{t}}_F^2 
&\le& \norm{\Mtilde_{t-1} - \Mtilde_t}_F^2
\cr
&\le& \norm{\Portho\( Z_{t-1} - Z_t \)}_F^2
\cr
&\le& \norm{Z_{t}-Z_{t-1}}_F^2 
\end{eqnarray}
for all $t\ge1$.
This proves that the sequence $\{\norm{Z_{t+1}-Z_{t}}_F^2\}$ converges (since it is decreasing and bounded below).

The convergence of $\{\norm{Z_{t+1}-Z_{t}}_F^2\}$ gives 
$$\norm{Z_{t+1}-Z_{t}}_F^2 - \norm{Z_{t}-Z_{t-1}}_F^2 \to 0 \text{ as } t\to\infty.$$
Then, by \eqref{lem:z-z<z-z},
\begin{eqnarray*}
0
&\ge& \norm{\Portho\( Z_{t}-Z_{t-1} \)}_F^2 - \norm{Z_{t}-Z_{t-1}}_F^2 
\cr 
&\ge& \norm{Z_{t+1}-Z_{t}}_F^2 - \norm{Z_{t}-Z_{t-1}}_F^2 
\cr 
&\to& 0 \quad\text{as } t \to \infty,
\end{eqnarray*}
which implies
\begin{eqnarray} \label{Pitself}
\norm{\Portho\( Z_{t}-Z_{t-1} \)}_F^2 - \norm{Z_{t}-Z_{t-1}}_F^2 \to 0 \Rightarrow \norm{\P\( Z_{t}-Z_{t-1} \)}_F^2 \to 0.
\end{eqnarray}

Furthermore, similarly to the proof of Lemma 2 in \cite{mazumder2010}, we can show
\begin{eqnarray} \label{fQf}
 f_{\tau}(Z_{t}) 
\ge Q_{\tau}(Z_{t+1}|Z_t)
\ge Q_{\tau}(Z_{t+1}|Z_{t+1})
= f_{\tau}(Z_{t+1}) \ge0
\end{eqnarray}
for every fixed $\tau_1,\ldots,\tau_\d >0$ and $t\ge 1$.
Thus, we have
$$ Q_{\tau}(Z_{t+1}|Z_t) - Q_{\tau}(Z_{t+1}|Z_{t+1}) \to 0 \quad\text{as } t\to\infty,$$
which implies 
$$\norm{\Portho\( Z_{t}-Z_{t+1} \)}_F^2 \to 0 \quad\text{as } t\to\infty.$$
The above along with \eqref{Pitself} gives 
$$\norm{Z_{t+1}-Z_{t}}_F^2 \to 0 \quad\text{as } t\to\infty.$$
\end{proof}

\begin{proof}[Proof of Theorem \ref{thm:Zsolvesf}]
By the construction of $D_{t}$, we have 
$$0 =  \(\Mtilde_{t} - Z_{t+1}\)-D_{t} \quad \text{for all } t\ge1.$$
Since $Z_\infty$ is a limit point of the sequence $Z_t$, there exists a subsequence $\{n_t\} \subset \{1,2,\ldots\}$ such that $Z_{n_t}\to Z_\infty$ as $t\to\infty$.
By Lemma \ref{z-zgoto0}, this subsequence $Z_{n_t}$ satisfies
$$Z_{n_t} - Z_{n_t +1}\to0$$
which implies
$$\Portho(Z_{n_t}) - Z_{n_t+1}\to \Portho(Z_\infty) - Z_\infty=-\P(Z_\infty).$$
Hence, 
\begin{eqnarray}\label{eq 01 : Thm conv}
	D_{n_t} = \(\P(\Mp)+\Portho(Z_{n_t})\)  - Z_{n_t+1} \to \P(\Mp)-\P(Z_\infty) = D_\infty.
\end{eqnarray}
 
Due to \eqref{assume4} and \eqref{eq 01 : Thm conv}, we have
\begin{eqnarray*}
	f_\tau ( Z^s) &\geq& f_\tau (Z_\infty) - \frac{1}{\n\d} \langle Z^s - Z_\infty, \P(\Mp)-\P(Z_\infty)-  D _\infty \rangle  \cr
	&=&f_\tau (Z_\infty).
\end{eqnarray*}
Since $f_\tau ( Z^s) \le f_\tau (Z_\infty)$ by definition of $Z^s$, we have $f_\tau ( Z^s) = f_\tau (Z_\infty)$.
Lastly, by \eqref{fQf}, we have $\lim_{t\rightarrow \infty} f_{\tau}(Z_{t}) = f(Z^s)$.
\end{proof}

\subsection{Proofs of Lemmas \ref{lem:undateTuning}-\ref{corol:assume2}}
\begin{proof}[Proof of Lemma \ref{lem:undateTuning}]
For $i=1,\ldots,\rank$, we have
\begin{eqnarray*}
&&\left| \frac{\tau_{t,i}}{\sqrt{\n\d}} - \frac{\tau_{t+1,i}}{\sqrt{\n\d}} \right|
\cr
&&= \frac{1}{\sqrt{\n\d}}\left|\blam_i (\Mtilde_t) - \sqrt{\blam_i^2 (\Mtilde_t) - \alphatilde_{t}} 
	- \lam_i (\Mtilde_{t+1}) + \sqrt{\blam_i^2 (\Mtilde_{t+1}) - \alphatilde_{t+1}}\right|
\cr
&&\le \frac{1}{\sqrt{\n\d}}\left|\blam_i (\Mtilde_t) - \(\sqrt{\lam_i^2 - \n\p\sig^2}\) \right|
	+ \frac{1}{\sqrt{\n\d}}\left|\sqrt{\blam_i^2 (\Mtilde_t) - \alphatilde_{t}} - \lam_i^2 \right|
\cr
&&\quad
	+ \frac{1}{\sqrt{\n\d}}\left|\blam_i (\Mtilde_{t+1}) - \(\sqrt{\lam_i^2 - \n\p\sig^2}\) \right|
	+ \frac{1}{\sqrt{\n\d}}\left|\sqrt{\blam_i^2 (\Mtilde_{t+1}) - \alphatilde_{t+1}} - \lam_i^2 \right|
\cr
&&= (I)+(II)+(III)+(IV).
\end{eqnarray*}
Then, by \eqref{lam2-lam2} and \eqref{rhohat-npsig2}, we have
\begin{eqnarray*}
(I) &=& \frac{1}{\sqrt{\n\d}}\left|\blam_i (\Mtilde_t) - \(\sqrt{\lam_i^2 - \n\p\sig^2}\) \right|
\cr
&=&\frac{1}{2\lam_\ast\sqrt{\n\d}}\left|\blam_i^2 (\Mtilde_t) - \(\lam_i^2 - \n\p\sig^2\) \right|
\cr
&\le&\frac{1}{2\lam_\ast\sqrt{\n\d}}\left| \blam_i^2(\Mtilde_t) - \alphatilde_{t} - \lam_i^2 \right|
	+ \frac{1}{2\lam_\ast\sqrt{\n\d}}\left| \alphatilde_{t} - \n\p\sig^2 \right|
\cr
&=&o_p\(\sqrt{\frac{h_\n}{\p\d}}\),
\end{eqnarray*}
where the second equality holds for some $\lam_\ast$ between $\blam_i (\Mtilde_t)$ and $\sqrt{\lam_i^2 - \n\p\sig^2}$
	by Taylor's expansion.
We can similarly show that $(III)=o_p\(\sqrt{h_\n/\p\d}\)$.
Both of $(II)$ and $(IV)$ are also $o_p\(\sqrt{h_\n/\p\d}\)$ by \eqref{lam-lam} and \eqref{lam2-lam2}.
\end{proof}

\begin{proof}[Proof of Lemma \ref{corol:assume2}]
 From Theorem \ref{thm:mathInduct} and the construction of $D_t$ in Assumption \ref{assume2}, we have
\begin{eqnarray*}
&& \left|\frac{1}{nd}\langle D_t - D_{t+1},Z_{t+1} - Z_{t+2} \rangle \right|
\cr
&&\le \frac{1}{nd} \norm{D_t - D_{t+1}}_F   \norm{Z_{t+1} - Z_{t+2}}_F
\cr
&&\le \frac{1}{nd} \norm{\Mtilde_{t} - Z_{t+1} - \(\Mtilde_{t+1} - Z_{t+2}\) }_F  \norm{Z_{t+1} - Z_{t+2}}_F
\cr
&&\le  \frac{1}{nd} \Big\{ \norm{\Mtilde_{t} - \Mtilde_{t+1}}_F + \norm{ Z_{t+1} - Z_{t+2} }_F \Big\}  \norm{Z_{t+1} - Z_{t+2}}_F
\cr
&&=  \frac{1}{nd} \Big\{ \norm{\Portho\(Z_t - Z_{t+1}\) }_F + \norm{ Z_{t+1} - Z_{t+2} }_F \Big\}   \norm{Z_{t+1} - Z_{t+2}}_F
\cr
&&\le \frac{1}{nd} \Big\{ \norm{Z_t - Z_{t+1}}_F + \norm{ Z_{t+1} - Z_{t+2} }_F \Big\}   \norm{Z_{t+1} - Z_{t+2}}_F
\cr
&&\le \frac{1}{nd} \Big\{ \norm{Z_{t} - \M}_F + 2\norm{Z_{t+1} - \M}_F + \norm{Z_{t+2} - \M}_F \Big\}
\cr
&&    \quad  \times \Big\{ \norm{Z_{t+1} - \M}_F + \norm{Z_{t+2} - \M}_F \Big\}
\cr
&&=o_p\(\frac{h_\n}{\p\d}\).
\end{eqnarray*}
\end{proof}










\bibliographystyle{Chicago}

\bibliography{Bibliography-MM-MC}
\end{document}